\documentclass[journal]{IEEEtran}

\usepackage{graphicx}
\usepackage{citesort}
\usepackage{amssymb}
\usepackage{amsfonts}
\usepackage{amsmath}
\usepackage{epsfig}
\usepackage{color}
\usepackage{fancybox}
\usepackage{textcomp}
\usepackage{multirow}
\usepackage{setspa ce}
\usepackage{psfrag}

\usepackage[ruled,vlined]{algorithm2e}

\usepackage{mathrsfs}

\newtheorem{Proposition}{Proposition}

\newtheorem{Definition}{Definition}
\newtheorem{Remark}{Remark}

    \def\ang#1{\mbox{$\langle #1 \rangle$}}

    \newcommand{\qh}{{\bf h}}

    \newcommand{\qx}{{\bf x}}
    \newcommand{\qy}{{\bf y}}
    \newcommand{\qz}{{\bf z}}

    \newcommand{\qA}{{\bf A}}

    \newcommand{\qH}{{\bf H}}
    \newcommand{\qI}{{\bf I}}

    \newcommand{\qP}{{\bf P}}
    \newcommand{\qQ}{{\bf Q}}
    \newcommand{\qR}{{\bf R}}
    \newcommand{\qS}{{\bf S}}

    \newcommand{\qW}{{\bf W}}
    \newcommand{\qX}{{\bf X}}
    \newcommand{\qY}{{\bf Y}}
    \newcommand{\qZ}{{\bf Z}}
    
    \newcommand{\qone}{{\bf 1}}

    \newcommand{\qOmega}{{\boldsymbol \Omega}}

    \newcommand{\qmu}{{\boldsymbol \mu}}

    \newcommand{\tc}{{\tilde{c}}}

    \newcommand{\tq}{{\tilde{q}}}

    \newcommand{\hr}{{\hat{r}}}
    \newcommand{\hs}{{\hat{s}}}
    \newcommand{\hq}{{\hat{q}}}
    \newcommand{\hp}{{\hat{p}}}
    \newcommand{\bp}{{\bar{p}}}
    \newcommand{\bv}{{\bar{v}}}

    \newcommand{\whqX}{{\widehat{\qX}}}
    \newcommand{\whqZ}{{\widehat{\qZ}}}
    \newcommand{\whqH}{{\widehat{\qH}}}
    \newcommand{\whX}{{\widehat{X}}}
    \newcommand{\whH}{{\widehat{H}}}
    \newcommand{\whZ}{{\widehat{Z}}}

    \newcommand{\wtqY}{{\widetilde{\qY}}}
    \newcommand{\wtY}{{\widetilde{Y}}}
    \newcommand{\wty}{{\widetilde{y}}}

    \newcommand{\uqX}{{\underline{\qX}}}

    \newcommand{\bbC}{{\mathbb C}}

    \newcommand{\calF}{{\mathcal F}}

    \newcommand{\calM}{{\mathcal M}}
    \newcommand{\calN}{{\mathcal N}}
    
    \newcommand{\calT}{{\mathcal T}}

    \newcommand{\calX}{{\mathcal X}}

    \newcommand{\calqH}{\boldsymbol{\cal H}}
    \newcommand{\calqX}{\boldsymbol{\cal X}}
    \newcommand{\calqQ}{\boldsymbol{\cal Q}}
    
    \newcommand{\calqZ}{\boldsymbol{\cal Z}}

    \newcommand{\tr}{{\sf tr}}
    \newcommand{\Ex}{{\sf E}}
    
    \newcommand{\Varx}{{\sf Var}}
    \newcommand{\Extr}{\operatornamewithlimits{\sf Extr}}
    
    \newcommand{\mse}{{\sf mse}}

    \newcommand{\sfQ}{{\sf Q}}
    \newcommand{\sign}{{\sf sign}}

    \newcommand{\rmd}{{\rm d}}
    \newcommand{\rmD}{{\rm D}}

    \newcommand{\sfb}{{\sf b}}
    \newcommand{\sfc}{{\sf c}}
    \newcommand{\sfd}{{\sf d}}

    \newcommand{\sfj}{{\sf j}}

    \newcommand{\sfs}{{\sf s}}
    \newcommand{\sft}{{\sf t}}

    \newcommand{\sfB}{{\sf B}}
    
    \newcommand{\sfF}{{\sf F}}
    \newcommand{\sfH}{{\sf H}}
    \newcommand{\sfP}{{\sf P}}
    \newcommand{\sfX}{{\sf X}}
    \newcommand{\sfZ}{{\sf Z}}

    \newcommand{\scP}{\mathscr{P}}

\begin{document}

\title{Bayes-Optimal Joint Channel-and-Data Estimation for Massive MIMO with Low-Precision ADCs}

\author{Chao-Kai~Wen,~Chang-Jen~Wang,~Shi~Jin,~Kai-Kit~Wong,~and~Pangan~Ting\thanks{C.-K. Wen and C.-J. Wang are with the Institute of Communications Engineering, National Sun Yat-sen University, Kaohsiung, Taiwan (e-mail: $\rm chaokai.wen@mail.nsysu.edu.tw$). S. Jin is with the National Mobile Communications Research Laboratory, Southeast University, Nanjing 210096, P. R. China. K. Wong is with the Department of Electronic and Electrical Engineering, University College London, United Kingdom. P. Ting is with the Industrial Technology Research Institute, Hsinchu, Taiwan.}}


\maketitle

\begin{abstract}
This paper considers a multiple-input multiple-output (MIMO) receiver with very low-precision analog-to-digital convertors (ADCs) with the goal of developing
massive MIMO antenna systems that require minimal cost and power. Previous studies demonstrated that the training duration should be {\em relatively long} to
obtain acceptable channel state information. To address this requirement, we adopt a joint channel-and-data (JCD) estimation method based on Bayes-optimal
inference. This method yields minimal mean square errors with respect to the channels and payload data. We develop a Bayes-optimal JCD estimator using a recent
technique based on approximate message passing. We then present an analytical framework to study the theoretical performance of the estimator in the large-system
limit. Simulation results confirm our analytical results, which allow the efficient evaluation of the performance of quantized massive MIMO systems and provide
insights into effective system design.

\end{abstract}


\section*{I. Introduction}
Fifth-generation (5G) mobile communication systems are expected to achieve a $1,000$-fold increase in capacity, a $10$-fold increase in spectral and energy efficiencies, and a $25$-fold increase in average cell throughput \cite{Wang-14COM-Mag}. Such significant enhancements can be achieved with large-scale multiple-input multiple-output (MIMO) antenna systems, which are also referred to as ``massive MIMO'' systems, e.g., \cite{Marzetta-10TW,Larsson-14COMMag,Andrews-14JSAC,Wang-14COM-Mag}. These systems employ hundreds, or even thousands, of antennas at base stations (BSs) to serve tens or hundreds of user terminals with the same time–frequency resources. As such, array gains are expected to grow infinitely with the number of antennas at the BSs, in which case the radiated energy efficiency increases dramatically and multiuser interference is eliminated completely.

However, the high dimensionality of massive MIMO systems considerably increases hardware cost and power consumption. In particular, the hardware complexity and power consumption of analog-to-digital converters (ADCs) increase exponentially with the number of bits per sample \cite{Walden-99JSAC} and thus present a major obstacle. This drawback has
motivated the use of low-cost low-precision ADCs (e.g., $1\text{-}3$ bits) for antennas, which has resulted in \emph{quantized} MIMO systems.\footnote{ADCs with a typical precision of $8\text{-}12$ bits are used in modern communication systems to process received signals in the digital domain. In this paper, the ``quantized'' MIMO system specifically represents a MIMO system equipped with very low-precision ADCs (e.g., $1\text{-}3$ bits).} Such
coarse quantization leads to the failure of all communication theories, as well as signal processing techniques dedicated to high-resolution quantization
\cite{Singh-09TCOM,Koch-10,Zhang-12TCOM,Wang-13TCOM}. Some aspects of quantized MIMO systems have been studied in the literature on capacity analysis
\cite{Mezghani-08ISIT,Mo-15TSP,Liang-15ArXiv}, energy efficiency analysis \cite{Bai-15TETT,Orhan-15ITA}, feedback codebook design \cite{Mo-15ArXiv-Feedback}, data detection \cite{Nakamura-08ISITA,Mezghani-10ISIT,Risi-14ArXiv,Wang-15TWCOM,Jacobsson-15ArXiv,Choi-15TSP,Xu-14JSM,Choi-15Arxiv,Studer-15Arxiv}, and channel estimation \cite{Mezghani-10WSA,Risi-14ArXiv,Mo-14ACSSP,Jacobsson-15ArXiv,Choi-15Arxiv,Studer-15Arxiv}.

The current work is focused on data detection and channel estimation for quantized MIMO systems. Previous studies on these subject mainly assumed perfect channel state information (CSI) at the receiver (CSIR) or considered problems in channel estimation. The use of coarse quantization greatly reduces the number of \emph{effective} measurements and hinders the easier acquisition of CSIR in quantized MIMO systems than in unquantized ones. As explained in \cite{Risi-14ArXiv}, a one-bit quantized MIMO system requires an extremely long training sequence (e.g., approximately $50$ times the number of users) to achieve the same performance as that in a full CSI case. This requirement motivates us to consider joint channel-and-data (JCD) estimation, in which the estimated payload data are utilized to aid channel estimation. A major advantage of JCD estimation is that relatively few pilot symbols are required to achieve equivalent channel and data estimation performances \cite{Takeuchi-13TIT,Ma-14TSP}.

Although an improved performance with the JCD technique is expected, its performance in \emph{quantized} MIMO systems is not clearly understood.\footnote{In the context of an \emph{unquantized} MIMO system, several aspects of the JCD estimation have been widely studied, see e.g., \cite{Takeuchi-13TIT,Ma-14TSP}.} The most related work appears to be that in \cite{Jacobsson-15ArXiv}, which investigated the achievable throughput in a one-bit quantized single-input single-output (SISO) channel using JCD estimation (i.e., least squares channel estimation jointly on pilot and data symbols). For the one-bit quantized MIMO system in \cite{Jacobsson-15ArXiv}, the authors considered a pilot-only scheme with least-squares channel estimation, followed by data detection that utilizes the maximal-ratio combining. Although high-order constellation, such as 16-QAM, was found to be capable of being supported by the one-bit quantized MIMO system, which outperforms the ones reported in \cite{Risi-14ArXiv} for QPSK, the long training sequence is still a requirement. Hence, the fundamental performance limits on quantized MIMO systems imposed by the JCD estimation represents an interesting research topic.

In the present work, we propose a framework for analyzing the achievable performance of quantized MIMO systems with JCD estimation. Unlike other JCD estimation schemes based on suboptimal criteria \cite{Jacobsson-15ArXiv,Takeuchi-13TIT,Ma-14TSP}, the Bayes-optimal inference for JCD estimation is used in this work because this approach generates minimum mean square errors (MMSEs) with respect to (w.r.t.) the channels and data symbols. In the conference version of this work \cite{Wen-15ISIT}, our simulation results indicate that the Bayes-optimal JCD estimator exhibits a significant advantage over pilot-only schemes in quantized MIMO systems. In addition to the derivations omitted in \cite{Wen-15ISIT}, the main contributions of this work are summarized as follows.
\begin{itemize}

    \item To implement the Bayes-optimal JCD estimator, we use a variant of belief propagation (BP) in approximating the marginal distributions of each data and channel component. We modify the bilinear generalized approximate message passing (BiG-AMP) algorithm in \cite{Parker-14TSP} and adapt it to the quantized MIMO system by providing the corresponding closed-form expressions for the nonlinear steps. We refer to this scheme as the GAMP-based JCD algorithm.\footnote{In this paper, the Bayes-optimal JCD estimator is regarded as the \emph{theoretical} optimal estimator, whereas the GAMP-based JCD algorithm is regarded as a \emph{practical method} for approximating the theoretical optimal estimator.}

    \item By performing a large-system analysis based on the replica method from statistical physics, we show the {\em decoupling principle} for the Bayes-optimal JCD estimator. That is, in the large-system regime, the input output relationship of a quantized MIMO system using the Bayes-optimal JCD estimator is decoupled into a bank of scalar additive white Gaussian noise (AWGN) channels w.r.t.~the data symbols and channel response. This decoupling property allows the characterization of several system performances of interest in an intuitive manner. In particular, the average symbol error rate (SER) w.r.t.~the data symbols and the average MSE w.r.t.~the channel estimate for the Bayes-optimal JCD estimator are determined.

    \item Finally, computer simulations are performed to verify the efficiency of the proposed GAMP-based JCD algorithm and the accuracy of our analysis. The high accuracy of our results ensures the quick and efficient evaluation of the performances of quantized MIMO systems. Several useful observations related to system design are derived from the analysis.
\end{itemize}

{\em Notations}---Throughout, for any matrix $\qA$, $A_{ij}$ refers to the $(i,j)$th entry of $\qA$, $\qA^T$ denotes the transpose of $\qA$, $\qA^H$ is the conjugate transpose of $\qA$, and $\tr(\qA)$ denotes its trace. Also, $\qI$ denotes the identity matrix, ${\bf 0}$ is the zero matrix, $\|\cdot\|_{\sfF}$ denotes the Frobenius norm, $\Ex[\cdot]$ represents the expectation operator, ${\rm log}(\cdot)$ is the natural logarithm, and $\sign(\cdot)$ is the signum function. In addition, a random vector $\qz$ drawn from the proper complex Gaussian distribution of mean $\qmu$ and covariance $\qOmega$ is described by the probability density function: 
\begin{equation*}
 \calN_{\bbC}(\qz;\qmu,\qOmega) = \frac{1}{\det(\pi \qOmega)} e^{-(\qz-\qmu)^H\qOmega^{-1}(\qz-\qmu)},
\end{equation*}
where $\det(\cdot)$ returns the determinant. We write $\qz \sim\calN_{\bbC}(\qz;\qmu,\qOmega)$. With $\rmD \qz$ denoting the real (or complex)  Gaussian integration measure, for an $n \times 1$ real valued vector $\qz$, we have
\begin{equation*}
\rmD\qz = \prod_{i=1}^n \phi(z_i)\, \rmd z_i \mbox{~~with~~} \phi(z_i) \triangleq \frac{e^{-\frac{z_i^2}{2}}}{\sqrt{2\pi}};
\end{equation*}
or ${\rmD\qz = \prod_{i=1}^n  \frac{e^{-({\rm Re}(z_i))^2-({\rm Im}(z_i))^2} }{\pi}} \rmd {\rm Re}(z_i) \rmd {\rm Im}(z_i)$ for the complex valued vector, where
${\rm Re}(\cdot)$ and ${\rm Im}(\cdot)$ extracts the real and imaginary components, respectively. Finally,
\begin{equation*}
    \Phi(x) \triangleq \frac{1}{\sqrt{2\pi}} \int_{-\infty}^x e^{-\frac{t^2}{2}} \, \rmd t .
\end{equation*}
denotes the cumulative Gaussian distribution function \cite{Hazewinkel-01}.

\begin{figure}
\begin{center}
\resizebox{3.5in}{!}{%
\includegraphics*{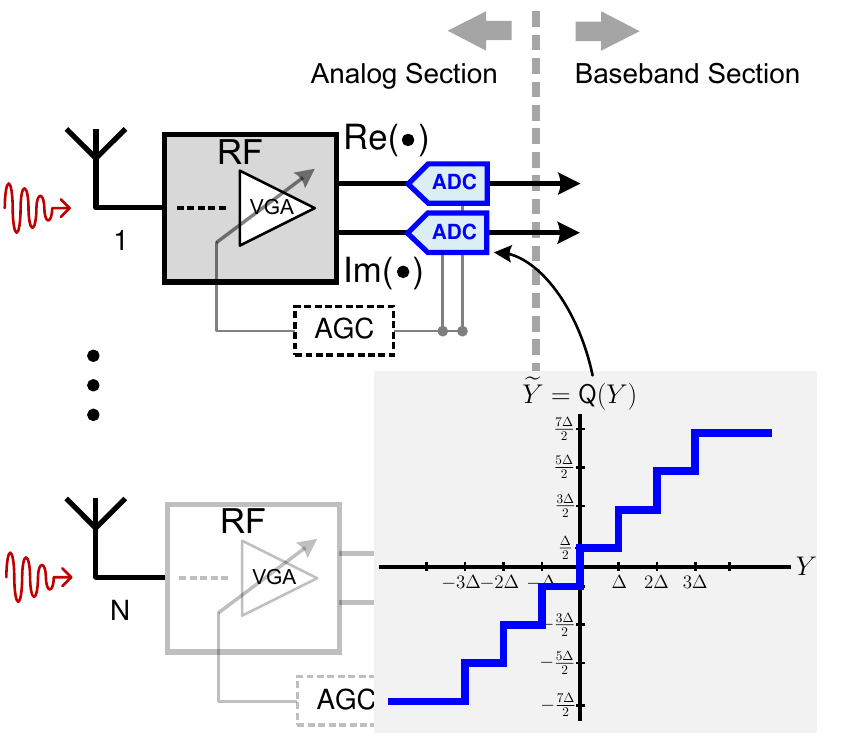} }%
\caption{The quantized MIMO antenna system.}\label{fig:QuantizedMIMO}
\end{center}
\end{figure}

\section*{II. System Model}
We consider a MIMO uplink system in which a BS is equipped with $N$ receive antennas serving $K$ single-antenna users. The channel is assumed to be flat block fading, wherein the channel remains constant over $T$ consecutive symbol intervals (i.e., a block). The received signal $\qY = [Y_{nt}] \in \bbC^{N \times T}$ over the block interval can be written in matrix form as
\begin{equation} \label{eq:sys}
    \qY = \frac{1}{\sqrt{K}}\qH \qX + \qW = \qZ + \qW,
\end{equation}
where $\qX = [X_{kn}] \in \bbC^{K \times T}$ denotes the transmit symbols in the block,  $\qH = [H_{nk}]\in \bbC^{N \times K}$ denotes the channel matrix containing the fading coefficients between the transmit antennas and the receive antennas, $\qW = [W_{nt}]\ \in \bbC^{N \times T}$ represents the additive temporally and spatially white Gaussian noise with zero mean and element-wise variance $\sigma_{w}^2$, and we define $\qZ = [Z_{nt}] \triangleq \frac{1}{\sqrt{K}}\qH \qX \in \bbC^{N \times T}$.

On the receiver side, each received signal is down-converted into analog baseband $Y_{nt}$ and then discretized using a \emph{complex-valued} quantizer $\sfQ_{c}$, as illustrated in Figure \ref{fig:QuantizedMIMO}. Here, each complex-valued quantizer $\sfQ_{c}(\cdot)$ is defined as $\widetilde{Y}_{nt} = \sfQ_{c}(Y_{nt}) \triangleq \sfQ({\rm Re}\{ Y_{nt}\}  ) + \sfj\sfQ({\rm Im}\{ Y_{nt} \} )$, i.e., the real and imaginary parts are quantized separately. In practice, a variable gain amplifier (VGA) with an automatic gain control (AGC) is used before the quantization to ensure that the analog baseband is within a proper range, e.g., ${(-1,+1)}$. Analog baseband $Y_{nt}$ is assumed to include the AGC gain and is thus in a proper range. The resulting quantized signal $\wtqY = [\wtY_{nt}] \in \bbC^{N \times T}$ is therefore given by
\begin{equation} \label{eq:qsys}
    \wtqY = \sfQ_{c}{\left(\qY\right)} = \sfQ_{c}{\left(\qZ + \qW\right)},
\end{equation}
where the quantization is applied element-wise.

Specifically, each complex-valued quantizer $\sfQ_{c}$ consists of two real-valued $\sfB$-bit quantizers $\sfQ$. Each real-valued quantizer maps a real-valued input to one of the $2^\sfB$ bins, which are characterized by the set of ${2^\sfB-1}$ thresholds $[r_1,r_2,\dots,r_{2^{\sfB}-1}]$, such that $ -\infty < r_1 < r_2 < \cdots <r_{2^{\sfB}-1} < \infty$. For notational convenience, we define $r_0 = -\infty$ and $r_{2^{\sfB}} = \infty$. The output is assigned a value in $(r_{b-1}, \,r_{b}]$ when the quantizer input falls in the interval $(r_{b-1}, \,r_{b}]$ (namely, the $b$-th bin). For example, the threshold of a typical uniform quantizer with the quantization step size $\Delta$ is given by
\begin{equation}
    r_b = {\left( -2^{\sfB-1} + b \right)} \Delta, \mbox{~~for~~} b=1,\dots,2^{\sfB}-1,
\end{equation}
and the quantization output is assigned the value ${r_b - \frac{\Delta}{2}}$ when the input falls in the $b$-th bin.\footnote{This output assignment is only true for ${b=1,\dots,2^{\sfB}-1}$. If $b=2^\sfB$, the quantization output is assigned the value ${\left( 2^{\sfB-1} - 2^{-1} \right)} \Delta$.} Figure \ref{fig:QuantizedMIMO} shows an example of the $3$-bit uniform quantizer. Notice that in practice, the VGA gain can be adjusted to attain the desired step size $\Delta$.

The channel matrix $\qH$ needs to be estimated at the receiver; thus, the first $T_{\sft}$ symbols of the block of $T$ symbols serve as the pilot sequences. The remaining $T_{\sfd}= T-T_{\sft}$ symbols are used for data transmissions. The training and data phases are referred to as $\sft$-phase and $\sfd$-phase, respectively. This setting is equivalent to partitioning $\qX$ and $\wtqY$ as
\begin{subequations}
\begin{align}
\qX &= \Big[\qX_{\sft}~\qX_{\sfd} \Big],~\mbox{with $\qX_{\sft}  \in \bbC^{K\times T_{\sft}} $, $\qX_{\sfd}  \in \bbC^{K \times T_{\sfd}} $,}\\
\wtqY &= \Big[\wtqY_{\sft}~\wtqY_{\sfd} \Big],~\mbox{with $\wtqY_{\sft}  \in \bbC^{N \times T_{\sft}} $, $\wtqY_{\sfd}  \in \bbC^{N \times T_{\sfd}} $}.
\end{align}
\end{subequations}
We assume that $\qX_{\sft}$ (or $\qX_{\sfd}$) is composed of independent and identically distributed (i.i.d.) random variables $\sfX_{\sft}$ (or $\sfX_{\sfd}$) drawn from the known probability distribution $\sfP_{\sfX_{\sft}}$ (or $\sfP_{\sfX_{\sfd}}$), i.e.,
\begin{equation} \label{eq:proX}
 \sfP_{\sfX}(\qX)
 = \underbrace{\Bigg(\prod_{k=1}^{K}\prod_{t=1}^{T_{\sft}} \sfP_{\sfX_{\sft}}{(X_{\sft,kt})}\Bigg)}_{ = \sfP_{\sfX_{\sft}}(\qX_{\sft})}
   \underbrace{\Bigg( \prod_{k=1}^{K}\prod_{t=1}^{T_{\sfd}} \sfP_{\sfX_{\sfd}}{(X_{\sfd,kt})} \Bigg)}_{ = \sfP_{\sfX_{\sfd}}(\qX_{\sfd})}.
\end{equation}
Given that the pilot and data symbols should appear on constellation points uniformly, the ensemble averages of $\{X_{\sft,kt}\}$ and $\{X_{\sfd,kt}\}$ are assumed to be zero. In addition, we let $\sigma_{x_{\sft}}^2$ and $\sigma_{x_{\sfd}}^2$ be the transmit powers in the $\sft$-phase and $\sfd$-phase, respectively, i.e., $\Ex\{|\sfX_{\sft,kt}|^2\} = \sigma_{x_{\sft}}^2$ and $\Ex\{|\sfX_{\sfd,kt}|^2\} = \sigma_{x_{\sfd}}^2$. For ease of notation, we  refer an entry of $\qX$ to $X_{kt}$ instead of $X_{\sft,kt}$ or $X_{\sfd,kt}$. Therefore, we use $ \calT_{\sft}=\{1,\ldots,T_{\sft} \}$ and $\calT_{\sfd}=\{T_{\sft}+1,\ldots,T \}$ to denote the sets of symbol indices in the $\sft$-phase and $\sfd$-phase, respectively.

Similarly, we assume that each entry $H_{nk}$ is drawn from a complex Gaussian distribution $\calN_{\bbC}(0,\sigma_{h}^2)$, where $\sigma_{h}^2$ is the large-scale fading coefficient. If $\sfP_{\sfH}{(H_{nk})} \equiv \calN_{\bbC}(0,\sigma_{h}^2)$, then
\begin{equation} \label{eq:proH}
 \sfP_{\sfH}(\qH) = \prod_{n=1}^{N}\prod_{k=1}^{K} \sfP_{\sfH}{(H_{nk})}.
\end{equation}
To prevent the key features of our results from being obfuscated by complex notations, we consider the case in which all the users have the same large-scale fading factor in the main text. A generalized version of our main result, in which the users have different large-scale fading factors, is presented in Appendix E. This generalization can be easily achieved by plugging user index $k$ into $\sigma_{h}^2$.

\section*{ III. Bayes-Optimal JCD Estimation}

We consider the case in which the receiver knows the distributions of $\qH$ and $\qX$ but not their realizations.
In the conventional pilot-only scheme, the receiver first uses $\wtqY_{\sft}$ and $\qX_{\sft}$ to generate an estimate of $\qH$ and then uses the estimated channel for estimating data $\qX_{\sfd}$ from $\wtqY_{\sfd}$ \cite{Risi-14ArXiv}.
In contrast to the pilot-only scheme, we consider JCD estimation, where the BS estimates both $\qH$ and $\qX_{\sfd}$ from $\widetilde{\qY}$ given $\qX_{\sft}$. We will treat the problem under the framework of Bayesian inference, which provides a foundation for achieving the best MSE estimates \cite{Poor-94BOOK}.

\subsection*{A. Theoretical Foundation}
We define the likelihood, i.e., the distribution of the received signals under (\ref{eq:qsys}) conditional on the unknown parameters, as
\begin{equation} \label{eq:lnkelihood}
    \sfP_{\sf out}(\wtqY|\qH,\qX) \triangleq \prod_{n=1}^{N} \prod_{t=1}^{T} {\sfP_{\sf out}{\left( \wtY_{nt} \Big| Z_{nt}\right)}},
\end{equation}
where
\begin{multline} \label{eq:lnkelihood_each}
    \sfP_{\sf out}{\left( \wtY \Big| Z \right)}
    = {\left( \frac{1}{\sqrt{ \pi \sigma_{w}^2}} \int_{r_{b-1}}^{r_b} e^{-\frac{(y - {\rm Re}(Z) )^2}{\sigma_{w}^2}} \,\rmd y \right)}   \\
      \times {\left( \frac{1}{\sqrt{ \pi \sigma_{w}^2}} \int_{r_{b'-1}}^{r_{b'}} e^{-\frac{(y - {\rm Im}(Z) )^2}{\sigma_{w}^2}} \,\rmd y  \right)}
\end{multline}
when ${ {\rm Re}(\wtY) \in (r_{b-1},r_{b}]}$ and ${{\rm Im}(\wtY) \in(r_{b'-1},r_{b'}]}$. Based on the cumulative Gaussian distribution function (see the definition in Notations), (\ref{eq:lnkelihood_each}) becomes
\begin{equation} \label{eq:lnkelihood_each_Simple}
    \sfP_{\sf out}{\left( \wtY \Big| Z \right)}
    = \Psi_b\big({\rm Re}(Z)\big) \Psi_{b'}\big({\rm Im}(Z)\big),
\end{equation}
where
\begin{equation} \label{eq:def_Psi_b}
    \Psi_b(x) \triangleq \Phi{\left( \frac{\sqrt{2}(r_b- x)}{\sigma_{w}}\right)} - \Phi{\left( \frac{\sqrt{2}(r_{b-1}- x)}{\sigma_{w}}\right)}.
\end{equation}
The prior distributions of $\qX$ and $\qH$ are given by (\ref{eq:proX}) and (\ref{eq:proH}), respectively. The posterior probability can then be computed according to Bayes' rule as
\begin{equation} \label{eq:posteriorPr}
    \sfP(\qH,\qX|\wtqY) = \frac{\sfP_{\sf out}(\wtqY|\qH,\qX)\sfP_{\sfH}(\qH)\sfP_{\sfX}(\qX)}{\sfP(\wtqY)},
\end{equation}
where 
\begin{equation} \label{eq:marginalPr}
    \sfP(\wtqY) = \int_{\qH} \int_{\qX}  \, \sfP(\wtqY|\qH,\qX)\sfP_{\sfH}(\qH)\sfP_{\sfX}(\qX) \, \rmd\qH \rmd\qX
\end{equation}
is the marginal likelihood.

Given the posterior probability, an estimator for $H_{nk}$ can be obtained by the posterior mean
\begin{equation} \label{eq:estH}
    \widehat{H}_{nk} = \int  H_{nk}\scP(H_{nk}) \, \rmd H_{nk},
\end{equation}
where $${\scP(H_{nk}) = \int_{\qH\setminus H_{nk}} \!\rmd\qH \int_{\qX}\!\rmd\qX  \, \sfP(\qH,\qX|\wtqY)} $$ denotes the marginal posterior probability of $H_{nk}$.
In (\ref{eq:estH}), the notation $\int_{\qH\setminus H_{nk}}\!\rm\rmd\qH $ denotes the integration over all the variables in $\qH$ except for $H_{nk}$. Similarly, the estimator
for $X_{o,kt}$ for $o \in \{ \sft,\sfd\}$ can be obtained by the posterior mean
\begin{equation} \label{eq:estX}
    \widehat{X}_{o,kt} = \int \scP(X_{o,kt}) X_{o,kt}\, \rmd X_{o,kt},
\end{equation}
where $$\scP(X_{o,kt}) = \int_{\qH}\!\rmd\qH \int_{\qX\setminus X_{o,kt}}\!\rmd\qX \, \sfP(\qH,\qX|\wtqY)$$ is the marginal posterior probability of $X_{o,kt}$. The notation $\int_{\qX\setminus X_{o,kt}} \!\rm\rmd\qX $ denotes the integration over all the variables in $\qX$ except for $X_{o,kt}$. The posterior mean estimators (\ref{eq:estH}) and (\ref{eq:estX}) minimize the (Bayesian) MSE \cite{Poor-94BOOK} defined as
\begin{subequations} \label{eq:mseHX}
\begin{align}
    \mse(\qH) &=  \frac{1}{NK}\Ex\left\{ \| \whqH - \qH\|_{\sfF}^2 \right\},
    \label{eq:mseH} \\
    \mse(\qX_{o}) &= \frac{1}{KT_{o}} \Ex{\left\{ \| \whqX_{o} - \qX_{o} \|_{\sfF}^2 \right\}}, ~\mbox{for}~ o \in \{ \sft,\sfd\},
    \label{eq:mseX}
\end{align}
\end{subequations}
where the expectation operator is w.r.t. $\sfP(\wtqY,\qH,\qX_{o})$; moreover, $\whqH = [\widehat{H}_{nk}]$ and $\whqX_{o} = [\widehat{X}_{o,kt}]$. We refer to (\ref{eq:estH}) and (\ref{eq:estX}) as the Bayes-optimal estimator.

\begin{Remark} \label{Remark_1}
Given a {\em known} pilot matrix, $\uqX_{\sft}$, which by definition is given by ${\sfP_{\sfX_{\sft}}(\qX_{\sft}) = \delta(\qX_{\sft}-\uqX_{\sft})}$, we obtain
$\widehat{X}_{\sft,kt} =\underline{X}_{\sft,kt}$ from (\ref{eq:estX}); therefore, $\mse(\qX_{\sft}) = 0$. For the case of interest, we always have $\mse(\qX_{\sft}) = 0$. The algorithm as well as the analytical results still work even if the pilots are unknown, and the MSE can be expressed as (\ref{eq:mseX}).
\end{Remark}

\subsection*{B. Bayes-Optimal Estimator in SISO Channel}
To better understand the Bayes-optimal estimator, we first explain it in a simple SISO system.
We consider a SISO version of the system (\ref{eq:sys}) given by
\begin{equation} \label{eq:example_sys}
 Y = Z + W.
\end{equation}
Recall that $W$ is the additive white Gaussian noise with zero mean and variance $\sigma_w^2$.
After a complex-valued quantizer, $\wtY = \sfQ_{c}(Y)$ is obtained. Based on the system model (\ref{eq:sys}), $Z = HX$ should be kept. However, to facilitate interpretation, we first let $Z$ be a random variable with distribution $\sfP_{\sfZ}$. According to Bayes' rule (\ref{eq:posteriorPr}), the
posterior probability can be computed as
\begin{equation} \label{eq:margPost_z}
    \sfP(Z|\wtY) = \frac{\sfP_{\sf out}(\wtY|Z)\sfP_{\sfZ}(Z)}{\sfP(\wtY)},
\end{equation}
where $\sfP(\wtY) = \int \sfP_{\sf out}(\wtY|z) \sfP_{\sfZ}(z) \,\rmd z $ is the marginal likelihood. Then, from (\ref{eq:estH}) or (\ref{eq:estX}), the posterior mean
estimator for $Z$ is given by
\begin{equation} \label{eq:postEst_z}
    \whZ = \int z\sfP(z|\wtY)  \,\rmd z.
\end{equation}

To specify the estimator, we further assume that $Z$ is a proper complex Gaussian with mean $\hp$ and variance $v^{p}$, i.e., $\sfP_{\sfZ}(Z) =
\calN_{\bbC}(Z;\hp,v^{p})$. Then, we derive the estimator (\ref{eq:postEst_z}) under the two channels, unquantized and quantized, in the following examples.

{\noindent {\bf Example~1} (Unquantized Channel).} In this case, we have $\wtY = Y$ and $\sfP_{\sf out}(\wtY|Z) = \frac{1}{\pi \sigma_w^2} e^{-|\wtY-Z|^2/\sigma_w^2}$. Using these distributions, we obtain
\begin{align}
    \sfP_{\sf out}(\wtY|Z)\sfP_{\sfZ}(Z) &= \calN_{\bbC}(Z;\wtY,\sigma_w^2) \calN_{\bbC}(Z;\hp,v^{p}) \notag \\
    &\hspace{-2.5cm}= D \cdot
    \calN_{\bbC}{\left(Z;\frac{ v^{p}\wtY + \sigma_w^2\hp }{ \sigma_w^2+v^{p} }, \frac{\sigma_w^2 v^{p} }{ \sigma_w^2+v^{p} } \right)}, \label{eq:post_Z_num}
\end{align}
where ${D = \calN_{\bbC}(0;\wtY-\hp,\sigma_w^2+v^{p}) }$, and the second equality follows the \emph{Gaussian reproduction property} \cite[(A.7)]{Rasmussen-06BOOK}.\footnote{The product of two Gaussians gives another Gaussian \cite[(A.7)]{Rasmussen-06BOOK}:
$$\mathcal{N}_{\bbC}(x;a,A)\mathcal{N}_{\bbC}(x;b,B)=D\cdot \mathcal{N}_{\bbC}(x;c,C),$$ where $c=C(A^{-1}a+B^{-1}b)$, $C=(A^{-1}+B^{-1})^{-1}$, and $D=\mathcal{N}_{\bbC}(0;a-b,A+B)$.} Substituting (\ref{eq:post_Z_num}) into (\ref{eq:margPost_z}), we obtain
\begin{equation} \label{eq:post_Z_final}
    \sfP(Z|\wtY) = \calN_{\bbC}{\left(Z;\frac{ v^{p}\wtY + \sigma_w^2\hp }{ \sigma_w^2+v^{p} }, \frac{\sigma_w^2 v^{p} }{ \sigma_w^2+v^{p} } \right)}.
\end{equation}
The estimator (\ref{eq:postEst_z}), which is the mean of $\sfP(Z|\wtY)$ after \emph{rearranging} is determined as
\begin{equation} \label{eq:postEst_Z}
    \whZ = \hp + \frac{ v^{p}}{ \sigma_w^2+v^{p} } (\wtY - \hp ).
\end{equation}
The MSE of the estimator, which is the variance of $\sfP(Z|\wtY)$, is
\begin{equation} \label{eq:postEst_Z_var}
    v^{z} = v^{p} - \frac{(v^{p})^2 }{ \sigma_w^2+v^{p} }.
\end{equation}

{\noindent {\bf Example~2} (Quantized Channel).} If ${ {\rm Re}(\wtY) \in (r_{b-1},r_{b}]}$ and ${{\rm Im}(\wtY) \in(r_{b'-1},r_{b'}]}$, then the likelihood of the quantized measurement $\wtY$ is given by (\ref{eq:lnkelihood_each_Simple}). The calculation of the posterior mean and variance in the quantized channel is technical, but it basically follows a procedure similar to that in the unquantized channel.
A derivation is given in Appendix A, which turns out to yield
\begin{align}
    \whZ
   &= \hp + \frac{\sign(\wtY) v^{p} }{\sqrt{2(\sigma_{w}^2 + v^{p})}} \left( \frac{\phi(\eta_1)-\phi(\eta_2)}{\Phi(\eta_1)-\Phi(\eta_2)} \right),
   \label{eq:hatZ_RealGaussian} \\
    v^{z}
      &= \frac{v^{p}}{2} - \frac{(v^{p})^2}{2(\sigma_{w}^2 + v^{p})}\times \nonumber \\
    & \hspace{0.5cm} \left( \frac{\eta_1\phi(\eta_1)-\eta_2\phi(\eta_2)}{\Phi(\eta_1)-\Phi(\eta_2)}
     + \left(\frac{\phi(\eta_1)-\phi(\eta_2)}{\Phi(\eta_1)-\Phi(\eta_2)}\right)^2 \right),
     \label{eq:mseZ_RealGaussian}
\end{align}
where
\begin{subequations} \label{eq:eta_def}
\begin{align}
    \eta_1 &= \frac{\sign(\wtY)\hp-\min\{|r_{b-1}|,|r_{b}|\}}{\sqrt{(\sigma_{w}^2 + v^{p})/2}}, \\
    \eta_2 &= \frac{\sign(\wtY)\hp-\max\{|r_{b-1}|,|r_{b}|\}}{\sqrt{(\sigma_{w}^2 + v^{p})/2}}.
\end{align}
\end{subequations}
The real and imaginary parts are quantized separately, and each complex-valued channel can be decoupled into two real-valued channels. The expressions (\ref{eq:hatZ_RealGaussian}) and (\ref{eq:mseZ_RealGaussian}) are the estimators only for the real part of $Z$. To facilitate notation, we have
abused $\wtY$ and $\whZ$ in (\ref{eq:hatZ_RealGaussian}) and (\ref{eq:mseZ_RealGaussian}) to denote ${\rm Re}(\wtY)$ and ${\rm Re}(\whZ)$, respectively. The estimator for the imaginary part ${\rm Im}(\whZ)$ can be obtained analogously as (\ref{eq:hatZ_RealGaussian}) and (\ref{eq:mseZ_RealGaussian}), while $\wtY$ and $b$ should be replaced by
${\rm Im}(\wtY)$ and $b'$, respectively.

\begin{Remark}
Recall ${r_{0} = -\infty}$ and ${r_{2^{\sfB}} = \infty}$. Therefore, if $b = 1$ or $b = 2^{\sfB}$, we obtain $\phi(\eta_2) = 0$, $\eta_2\phi(\eta_2) = 0$, and $\Phi(\eta_2) = 0$. Additionally, for the special case of $\sfB = 1$ (i.e., one-bit quantizer), the expressions of (\ref{eq:hatZ_RealGaussian}) and (\ref{eq:mseZ_RealGaussian}) agree with those reported in \cite{Ziniel-15TSP}.
\end{Remark}

\begin{Remark}
For another extreme case of $\sfB \rightarrow \infty$ and $\Delta \rightarrow 0$, we return to the unquantized channel in Example 1. Instead of using the procedure in Example 1, we show how the expressions (\ref{eq:postEst_Z}) and (\ref{eq:postEst_Z_var}) can be obtained from (\ref{eq:hatZ_RealGaussian}) and (\ref{eq:mseZ_RealGaussian}). Recall that $r_{b-1}$ and $r_{b}$ are the upper and lower bin boundary positions w.r.t. the $b$-th bin, respectively. Let $r_{b-1} = r$ and $r_{b} = r_{b-1} + \rmd r$. As $\sfB \rightarrow \infty$ and $\Delta \rightarrow 0$, we obtain $\rmd r \rightarrow 0$, which results in $r_{b} \rightarrow r$ and $\eta_1 \rightarrow \eta_2 \triangleq \eta$. Furthermore, we obtain ${\Phi(\eta_1)-\Phi(\eta_2) \rightarrow \frac{ \rmd }{ \rmd r } \Phi(\eta) }$, $\phi(\eta_1)-\phi(\eta_2) \rightarrow \frac{\rmd }{ \rmd r } \phi(\eta)$, and $\eta_1\phi(\eta_1)-\eta_2\phi(\eta_2) \rightarrow \frac{ \rmd  }{ \rmd r }\eta\phi(\eta)$. By substituting these relationships into (\ref{eq:hatZ_RealGaussian})--(\ref{eq:mseZ_RealGaussian}) and applying the facts that $\frac{\rmd }{ \rmd r } \Phi(\eta) = \phi(\eta) \frac{ \rmd}{ \rmd r } \eta $, $\frac{\rmd }{ \rmd r } \phi(\eta) = - \eta \phi(\eta) \frac{ \rmd }{ \rmd r } \eta$, and $\frac{\rmd }{ \rmd r }  \eta \phi(\eta)= (1 - \eta^2) \phi(\eta) \frac{ \rmd}{ \rmd r } \eta $, we recover the same expressions as those given in (\ref{eq:postEst_Z}) and (\ref{eq:postEst_Z_var}) for the real part of $\whZ$. The imaginary part for $\whZ$ can be obtained analogously.
\end{Remark}

The aforementioned example is the estimator for $Z$. The same concept can be easily applied to the estimate of $H$ or $X$, if $Z$ is replaced by $H$ or $X$ in (\ref{eq:example_sys}). However, if $Z = HX$ and both $H$ and $X$ are unknown, the complexity of the Bayes-optimal estimator increases. In this case, the posterior probability in (\ref{eq:margPost_z}) becomes $\sfP(H,X|\wtY) = \frac{\sfP_{\sf out}(\wtY|H,X)\sfP_{\sfH}(H)\sfP_{\sfX}(X)}{\sfP(\wtY)}$, which involves two prior distributions for $H$ and $X$ as that in (\ref{eq:posteriorPr}). To implement the posterior mean estimator for $H$ and $X$, we need the marginal posterior probabilities $\scP(H) = \int \sfP(H,X|\wtY) \, \rmd X$ and $\scP(X) = \int \sfP(H,X|\wtY) \, \rmd H$, respectively. A closed form for the posterior probability $\sfP(H,X|\wtY)$ does not appear possible. Although we can resort to numerical integration to implement the estimator, the computational complexity is high. Therefore, one might consider an alternative technique; that is, the estimate of $H$ is performed with fixed $X$ and vice versa.

\begin{figure}
\begin{center}
\resizebox{3.0in}{!}{%
\includegraphics*{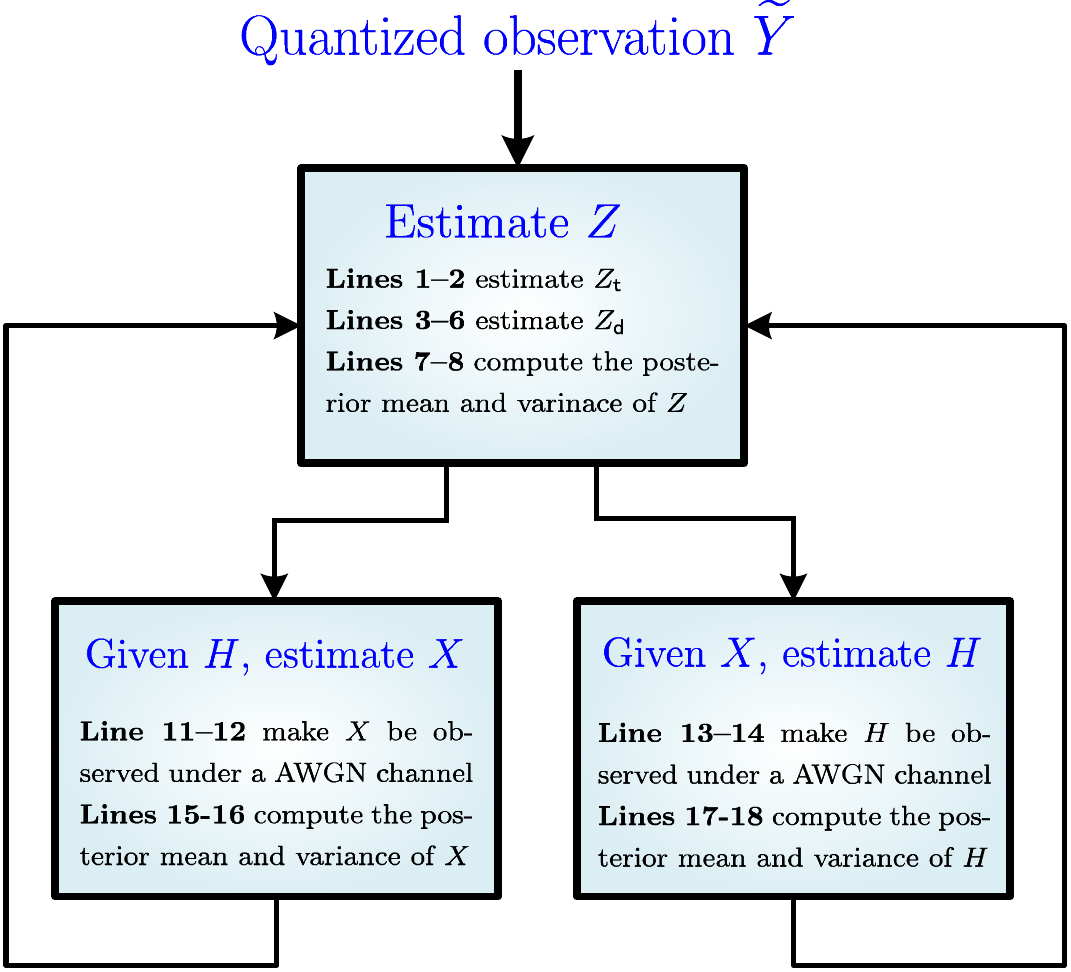} }%
\caption{A representation of the GAMP-based JCD algorithm.}\label{fig:GAMP-basedJCD}
\end{center}
\end{figure}

In the next subsection, we develop a practical algorithm for the Bayes-optimal estimator in the quantized MIMO system. Before proceeding, we intend to provide an intuition on the algorithm. A representation of the algorithm is shown in Fig.~\ref{fig:GAMP-basedJCD},
which seems to operate in the alternative manner. Conceptually, when the posterior mean and variance of $Z$ are obtained from the quantized observation $\wtY$, we can \emph{reconstruct} $Y$ and then approximate $\sfP_{\sf out}(Y|Z)$ as a Gaussian distribution. Then, the posterior mean estimator for $H$ (or $X$) can be conducted through $Y$, which is an \emph{AWGN} channel rather than a quantized channel, in an alternative manner. This representation is merely for an intuition. The accurate algorithm development takes a different route.

\subsection*{B. GAMP-Based JCD Algorithm}

\begin{figure}
\begin{center}
\resizebox{2.8in}{!}{%
\includegraphics*{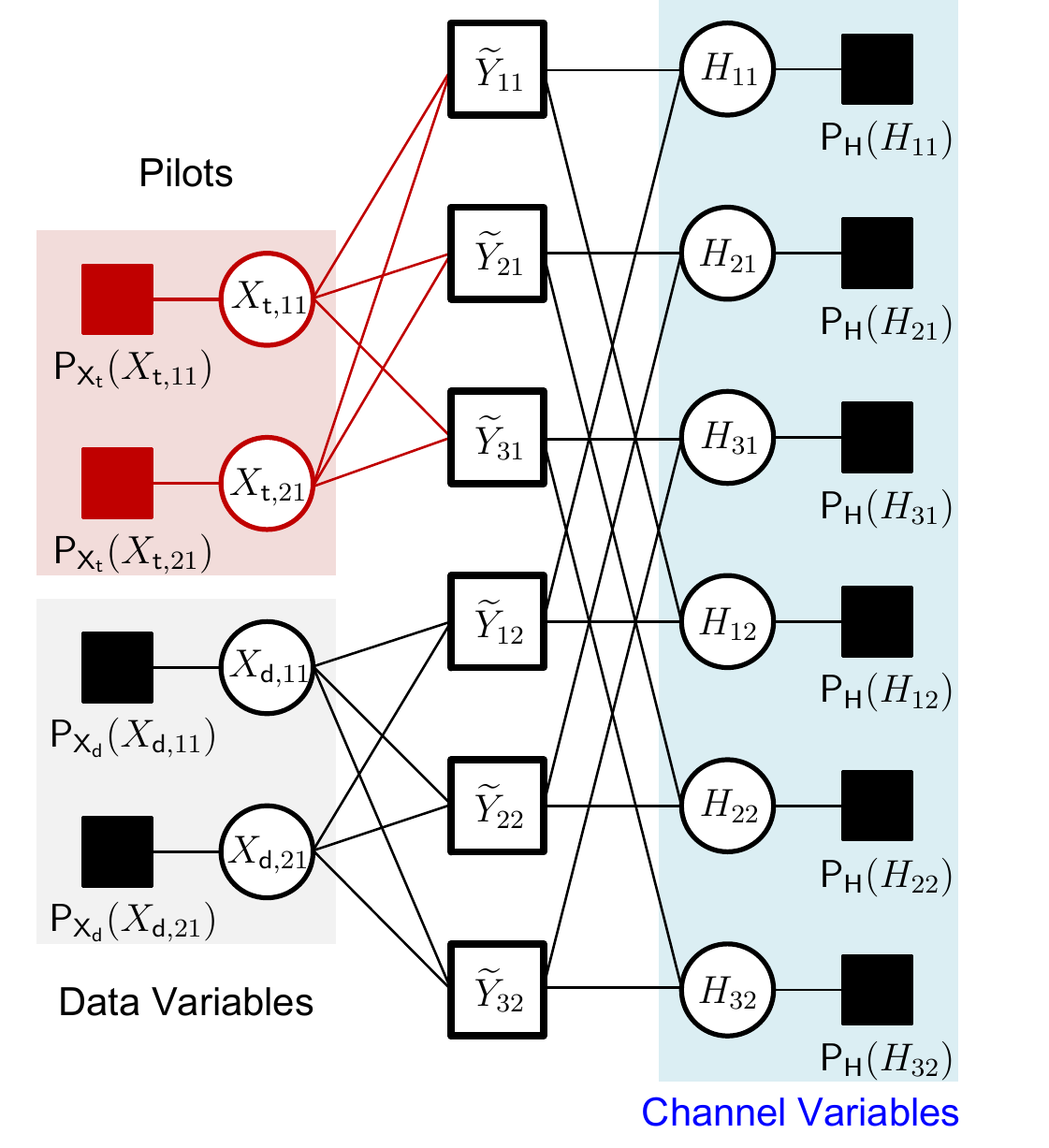} }%
\caption{Factor graph representation of the integrant of (\ref{eq:postFG}), where $N=3$, $K=2$, and $T=2$.}\label{fig:FactorGraph}
\end{center}
\end{figure}

From the discussions above, the direct computations of (\ref{eq:estH}) and (\ref{eq:estX}) are intractable because of the high-dimensional integrals in the marginal posteriors $\scP(X_{kt})$ and $\scP(H_{nk})$. To make these tractable, we first note that by combining (\ref{eq:proX})--(\ref{eq:lnkelihood}), the posterior probability (\ref{eq:posteriorPr}) can be factored into
\begin{multline} \label{eq:postFG}
    \frac{1}{\sfP(\wtqY)} \, \prod_{n=1}^{N} \prod_{t=1}^{T} {\sfP_{\sf out}{\left( \wtY_{nt} \Big| Z_{nt}\right)}}
    \times \prod_{n=1}^{N} \prod_{k=1}^{K} \sfP_{\sfH}{(H_{nk})} \\
    \times \prod_{k=1}^{K} \prod_{t=1}^{T_{\sft}} \sfP_{\sfX_{\sft}}{(X_{\sft,kt})}
    \times \prod_{k=1}^{K} \prod_{t=1}^{T_{\sfd}} \sfP_{\sfX_{\sfd}}{(X_{\sfd,kt})}.
\end{multline}
An example factor graph for (\ref{eq:postFG}) is shown in Figure~\ref{fig:FactorGraph}, where a square represents a factor node associated with the sub-constraint function $\sfP_{\sf out}{( \wtY_{nt} | Z_{nt})}$ in (\ref{eq:postFG}), whereas a circle shows a variable node associated with $H_{nk}$, $X_{\sfd,kt}$, or $X_{\sft,kt}$. The factor graph suggests the use of the canonical sum-product algorithm to compute the marginal posterior probabilities. The algorithm uses a set of message-passing equations that go from factor nodes to variable nodes and vice versa.

However, the computational complexity of the sum-product algorithm remains infeasible in the case of interest because it still involves high-dimensional integration and summation. Thus, we resort to recently developed approximation algorithms: the so-called AMP (approximate message-passing) algorithm \cite{Donoho-09PNAS} and the generalized AMP (GAMP) algorithm  \cite{Rangan-11ISIT}. In particular, the AMP algorithm, which is a variant of the sum-product algorithm, was initially proposed by Donoho {\em et al.}~\cite{Donoho-09PNAS} to solve a linear inverse problem in the context of compressive sensing. The use of GAMP in our MIMO system means that given $\qH$ is known; thus, GAMP can provide a tractable method to approximate the marginal posteriors $\scP(X_{kt})$'s. This part corresponds to addressing the message-passing equations between $\wtY_{nt}$ and $X_{kt}$, i.e., the left-hand side of Figure~\ref{fig:FactorGraph}. For the study, see \cite{Wang-15TWCOM,Wu-14JSTSP}. More recently, Parker {\em et al.}~in \cite{Parker-14TSP} applied a similar GAMP strategy, referred to as BiG-AMP, to the problem of reconstructing matrices from bilinear noisy observations (i.e., reconstructing $\qH$ and $\qX$ from $\qY$).

\begin{algorithm}[!h]\label{ago:BiGAMP-JCD}  \footnotesize
  \caption{ GAMP-based JCD Algorithm}
  \SetKwInOut{Input}{input}
  \SetKwInOut{Output}{output}
  \SetKwInOut{Initialize}{initialize}
  \SetKwInOut{Definition}{definition}

  \SetKwProg{Fn}{}{\string:}{}

  \Input{Quantized observations $\wtqY$, pilot matrix $\qX_{\sft}$, likelihood $\sfP_{\sf out}( Y | Z )$, and variable distributions $\sfP_{\sfH}(H)$ and $\sfP_{\sfX_{\sfd}}(X_{\sfd})$}
  \BlankLine
  \Output{$\whqH$, $\whqX_{\sfd}$ }
  \BlankLine
  \Definition{$\sum_{k} \triangleq \sum_{k=1}^{K}$, $\sum_{n} \triangleq \sum_{n=1}^{N}$}
  \BlankLine
  \Initialize{$\xi \leftarrow 1$;
  $\forall n,t$: $\hs_{nt}(0) = 0$, $v^{z}_{nt}(0) = 1$, $\whZ_{nt}(0) = 0$;
  $\forall n,k,t$: $v^{x}_{kt}(1) = 1$, $\whX_{\sfd,kt}(1) = 0$, $v^{h}_{nk}(1) = 1$, $\whH_{nk}(1) = 0$.
  }
  \BlankLine
  \vspace{0.1cm}
  \While{$\frac{\sum_{n,t} |\whZ_{nt}(\xi) - \whZ_{nt}(\xi-1)|^2}{ \sum_{n,t} |\whZ_{nt}(\xi-1)|^2 } > \epsilon$ {\rm \bf and} $\xi < \xi_{\max}$ }{
  \BlankLine
  \If{$t \in \calT_{\sft}$}{
  \nl $\forall n$: $v^{p}_{nt}(\xi) = \sum_{k} v_{nk}^{h}(\xi)|X_{kt}|^2 $\; \vspace{0.1cm}
  \nl $\forall n$: $\hp_{nt}(\xi) = \sum_{k} \whH_{nk}(\xi) X_{kt} - \hs_{nt}(\xi-1)v^{p}_{nt}(\xi)$\;
  }\BlankLine
  \If{$t \in \calT_{\sfd}$}{
  \nl $\forall n$: $\bv^{p}_{nt}(\xi) = \sum_{k} \big|\whH_{nk}(\xi)|^2v_{kt}^{x}(\xi) + v_{nk}^{h}(\xi)|\whX_{kt}(\xi)|^2$\; \vspace{0.1cm}
  \nl $\forall n$: $\bp_{nt}(\xi) = \sum_{k} \whH_{nk}(\xi) \whX_{kt}(\xi)$\; \vspace{0.1cm}
  \nl $\forall n$: $v^{p}_{nt}(\xi) = \bv^{p}_{nt}(\xi)+ \sum_{k} v^{h}_{nk}(\xi) v^{x}_{kt}(\xi) $\; \vspace{0.1cm}
  \nl $\forall n$: $\hp_{nt}(\xi) = \bp_{nt}(\xi) - \hs_{nt}(\xi-1)\bv^{p}_{nt}(\xi)$\;
  }\BlankLine
  \nl $\forall n,t$: $v^{z}_{nt}(\xi) = \Varx{\left\{ Z_{nt} \big|\hp_{nt}(\xi), v^{p}_{nt}(\xi)\right\}}$\; \vspace{0.1cm}
  \nl $\forall n,t$: $\whZ_{nt}(\xi) = \Ex{\left\{ Z_{nt} \big|\hp_{nt}(\xi), v^{p}_{nt}(\xi)\right\}}$\; \vspace{0.1cm}
  \vspace{0.2cm}

  \nl $\forall n,t$: $v^{s}_{nt}(\xi) = (1-v^{z}_{nt}(\xi)/v^{p}_{nt}(\xi))/v^{p}_{nt}(\xi)$\; \vspace{0.1cm}
  \nl $\forall n,t$: $\hs_{nt}(\xi) = (\whZ_{nt}(\xi)-\hp_{nt}(\xi))/v^{p}_{nt}(\xi)$\; \vspace{0.1cm}
  \vspace{0.2cm}

  \nl $\forall k,t$: $v^{r}_{kt}(\xi) = \left[\sum_{n}|\whH_{nk}(\xi)|^2 v^{s}_{nt}(\xi)\right]^{-1}$\; \vspace{0.1cm}
  \nl $\forall k,t$: $\hr_{kt}(\xi) = \whX_{kt}(\xi)\left(1-v^{r}_{kt}(\xi)\sum_{n}v^{h}_{nk}(\xi)v^{s}_{nt}(\xi) \right)$ \vspace{0.05cm}\\
  \hspace{1.85cm} $+ v^{r}_{kt}(\xi) \sum_{n}\whH_{nk}^{*}(\xi)\hs_{nt}(\xi) $ \; \vspace{0.1cm}

  \nl $\forall n,k$: $v^{q}_{nk}(\xi) = \Big[\sum_{t \in \calT_{\sft}}|X_{kt}|^2 v^{s}_{nt}(\xi)\Big]^{-1}$\; \vspace{0.1cm}
  \nl $\forall n,k$: $\hq_{nk}(\xi) = \whH_{nk}(\xi)\left(1-v^{q}_{nk}(\xi)\sum_{t \in \calT_{\sfd}}v^{x}_{kt}(\xi)v^{s}_{nt}(\xi) \right)$ \vspace{0.05cm}\\
  \hspace{2.0cm} $+ v^{q}_{nk}(\xi) \Big( \sum_{t \in \calT_{\sft}}X_{kt}^{*}\hs_{nt}(\xi)$ \vspace{0.05cm}\\
  \hspace{3.2cm} $+ \sum_{t \in \calT_{\sfd}}\whX_{kt}^{*}(\xi)\hs_{nt}(\xi) \Big)$ \; \vspace{0.1cm}
  \vspace{0.2cm}

  \nl $\forall k,t \in \calT_{\sfd}$: $v^{x}_{kt}(\xi+1) = \Varx{\left\{ X_{kt} \big|\hr_{kt}(\xi), v^{r}_{kt}(\xi)\right\}}$\; \vspace{0.1cm}
  \nl $\forall k,t \in \calT_{\sfd}$: $\whX_{kt}(\xi+1) = \Ex{\left\{ X_{kt} \big|\hr_{kt}(\xi), v^{r}_{kt}(\xi)\right\}}$\; \vspace{0.1cm}
  \vspace{0.2cm}

  \nl $\forall n,k$: $v^{h}_{nk}(\xi+1) = \Varx{\left\{ H_{nk} \big|\hq_{nk}(\xi), v^{q}_{nk}(\xi)\right\}}$\; \vspace{0.1cm}
  \nl $\forall n,k$: $\whH_{nk}(\xi+1) = \Ex{\left\{ H_{nk} \big|\hq_{nk}(\xi), v^{q}_{nk}(\xi)\right\}}$\; \vspace{0.1cm}

  \vspace{0.2cm}
  $\xi \leftarrow \xi+1$ \;
  }
\end{algorithm}

BiG-AMP for JCD estimation is presented in Algorithm \ref{ago:BiGAMP-JCD} for a given instantiation of the quantized observations $\wtqY$, the pilot matrix $\qX_{\sft}$, as well as the likelihood $\sfP_{\sf out}( \wtqY | \qZ )$, and the variable distributions $\sfP_{\sfH}(\qH)$ and $\sfP_{\sfX_{\sfd}}(\qX_{\sfd})$. We refer to this scheme as the GAMP-based JCD algorithm, which follows the same structure as BiG-AMP \cite{Parker-14TSP} except for the steps in dealing with the known pilots, i.e., $t \in \calT_{\sft}$ in Algorithm \ref{ago:BiGAMP-JCD}. Reference \cite{Parker-14TSP} presents the derivation details of BiG-AMP.

To better understand the algorithm, we provide some intuition on each step of Algorithm \ref{ago:BiGAMP-JCD}.
Fig.~\ref{fig:GAMP-basedJCD} illustrates the structure of the algorithm. In lines 3--6, an estimate ${\hat{\qP}_{\sfd} =[\hp_{nt}]}$ of the matrix product $\qZ_{\sfd} = \qH \qX_{\sfd}$ and the corresponding variances $\{ v^{p}_{nt} : {t \in \calT_{\sfd}} \}$ are estimated. ${\bar{\qP}_{\sfd} =[\bp_{nt}] }$ and $\{ \bv^{p}_{nt} : {t \in \calT_{\sfd}} \}$ in
lines 3 and 4 can be regarded as auxiliary variables.\footnote{$\bar{\qP}_{\sfd}$ is a plug-in estimate of $\qZ_{\sfd}$, whereas $\hat{\qP}_{\sfd} =[\hp_{nt}]$ provides a refinement by introducing the ``Onsager'' correction into the context of AMP. For details, see \cite{Parker-14TSP}.} The same procedure is followed in lines 1 and 2 but for the matrix product $\qZ_{\sft} = \qH \qX_{\sft}$. Given that the pilot matrix $\qX_{\sft}$ is known, the corresponding variances for $\qX_{\sft}$ are zero, i.e.,
$v^{x}_{kt} = 0$ for ${t \in \calT_{\sft}}$. With $v^{x}_{kt} = 0$, we thus have plugged $\bp_{nt}$ and $\bv^{p}_{nt}$ into $\hp_{nt}$ and $v^{p}_{nt}$ for ${t \in
\calT_{\sft}}$ to obtain lines 1 and 2. The posterior means $\whqZ =[\whZ_{nt}]$ and variances $\{ v^{z}_{nt} \}$
of $\qZ$ are obtained in lines 7 and 8 using $\{\hp_{nt},v^{p}_{nt}\}$. Subsequently, the posterior moments are used in lines 9 and 10 to compute the residual $\hat{\qS} =[\hs_{nt}]$ and the inverse residual variances $\{ v^{s}_{nt} \}$. In lines 11 and 12, these residual terms are used to compute ${\hat{\qR} = [\hr_{kt}]}$ and $\{ v^{r}_{kt} \}$, where $\hr_{kt}$ can be interpreted as an observation of $X_{\sfd,kt}$ under an AWGN channel with zero mean and a variance of $v^{r}_{kt}$. Similarly, $\hat{\qQ} = [\hq_{nk}]$ and $v^{q}_{nk}$, where $\hq_{nt}$ can be interpreted as an observation of $H_{nk}$ under an AWGN channel with noise variance of $v^{q}_{nk}$, are evaluated in lines 13 and 14. Finally, the posterior mean $\whqX = [\whX_{kt}]$ and variances $\{ v^{x}_{kt} \}$ are estimated in lines 15 and 16 by taking into account the prior $\sfP_{\sfX_{\sfd}}$; the same is performed for $H_{nk}$ in lines 17 and 18.

\subsection*{C. Nonlinear Steps}
Algorithm \ref{ago:BiGAMP-JCD} provides a high-level description of BiG-AMP to perform JCD estimation. Lines 7--8, 15--16, and 17--18 of Algorithm \ref{ago:BiGAMP-JCD} perform the posterior mean and variance estimators for $Z_{nt}$, $X_{kt}$, and $H_{nk}$, respectively.  A remarkable feature of the algorithm
is that at each iteration, the estimates of $Z_{nt}$, $X_{kt}$, and $H_{nk}$ can separately serve as the estimators over a bank of scalar channels. Next, we describe these nonlinear steps in detail. For brevity, we omit the subscript indexes $n,k,t$ hereafter.

First, we notice that lines 7--8 compute the posterior mean and variance of $Z$; in this computation, the expectation operator is w.r.t.
\begin{equation*} 
    \scP(Z) = \frac{{\sfP_{\sf out}( \wtY | Z )} {\sfP_{\sf Z}( Z )}}
    { \int {\sfP_{\sf out}( \wtY | z' )} \sfP_{\sf Z}( z' )\,\rmd z' },
\end{equation*}
where $ \sfP_{\sf out} ( \wtY | Z )$ is given by (\ref{eq:lnkelihood_each_Simple}), and ${\sfP_{\sf Z}( Z )} = \calN_{\bbC}(Z; \hp, v^{p})$. This process is identical to that implemented in Example 2. As a result, lines 7--8 of Algorithm \ref{ago:BiGAMP-JCD} for each real-valued channel can be computed using the expressions in (\ref{eq:hatZ_RealGaussian})--(\ref{eq:mseZ_RealGaussian}).

Next, we discuss the nonlinear steps used to compute $(\whX,v^{x})$ and $(\whH,v^{h})$ in lines 15--16 and 17--18 of Algorithm \ref{ago:BiGAMP-JCD}. Specifically, the expectations and variances in lines 15-16 and 17-18 are taken w.r.t.~the marginal posterior
\begin{align}
    \scP(X_{\sfd}) &=
    \frac{\calN_{\bbC}(X_{\sfd}; \hr, v^{r}) {\sfP_{\sfX_{\sfd}}( X_{\sfd} )}}
    { \int \calN_{\bbC}(x_{\sfd}'; \hr, v^{r}) {\sfP_{\sfX_{\sfd}}( x_{\sfd}' )}\,\rmd x_{\sfd}'}, \label{eq:margPost_x2}\\
    \scP(H) &=
    \frac{ \calN_{\bbC}(H; \hq, v^{q}) {\sfP_{\sfH}( H )}}
    { \int \calN_{\bbC}(h'; \hq, v^{q}) {\sfP_{\sfH}(h')}\,\rmd h' }.  \label{eq:margPost_h}
\end{align}
These posterior probabilities are similar to those of $(\whZ,v^{z})$, except that $\sfP_{\sf out}$ is replaced with a Gaussian distribution and the corresponding priors $\sfP_{\sfX_{\sfd}}$ and $\sfP_{\sfH}$ are used in place of $\sfP_{\sfZ}$. In fact, the former change results in an estimator that is fundamentally different from that
in the case of $Z$. In Example 1, if $\sfP_{\sf out}$ is a Gaussian distribution, then the estimator is operated in an \emph{unquantized} channel. That is, the estimates of $H$ and $X$ in Algorithm \ref{ago:BiGAMP-JCD} are based on AWGN channels.

To specify $(\whX,v^{x})$, we consider the square QAM constellation with ${2\nu \times 2\nu}$ points
\begin{align}
    \calX &= \Big\{ X_{\rm R}+\sfj X_{\rm I} : X_{\rm R},X_{\rm I} \in \left\{ -(2\nu-1)\zeta, \ldots, -3\zeta,-\zeta, \right. \nonumber \\
    & \hspace{3.2cm} \left. \zeta,3\zeta,\ldots , (2\nu-1)\zeta \right\} \Big\},
\end{align}
where $\zeta = 1/\sqrt{2((2\nu)^2-1)/3}$ is the power normalization factor.
If $X$ is drawn from the constellation points uniformly, i.e., $\sfP_{\sfX_{\sfd}}(X_{\sfd}) = 1/(2\nu)^2$ for $X_{\sfd} \in \calX$, then lines 15--16 of Algorithm \ref{ago:BiGAMP-JCD} can be computed using
\begin{align}
    \whX_{\sfd}
    &= \frac{\sum_{i=1}^{\nu}(2i-1)\zeta F_{i}^{\sfs}{\left( {\rm Re}(\hr) \right)} }
    {\sum_{i=1}^{\nu} F_{i}^{\sfc}{\left( {\rm Re}(\hr) \right)}} \nonumber \\
    & \hspace{0.5cm} +\sfj \frac{\sum_{i=1}^{\nu}(2i-1)\zeta F_{i}^{\sfs}{\left( {\rm Im}(\hr) \right)}}
    {\sum_{i=1}^{\nu} F_{i}^{\sfc}{\left( {\rm Im}(\vartheta_i) \right)}}, \label{eq:hatX_RealGaussian} \\
    v^{x}
      & = \frac{\sum_{i=1}^{\nu}(2i-1)^2\zeta^2 F_{i}^{\sfs}{\left( {\rm Re}(\hr) \right)} }
    {\sum_{i=1}^{\nu}F_{i}^{\sfc}{\left( {\rm Re}(\hr) \right)}}
    \nonumber \\
     & \hspace{0.5cm}  +
     \frac{\sum_{i=1}^{\nu}(2i-1)^2\zeta^2 F_{i}^{\sfs}{\left( {\rm Im}(\hr) \right)}}
    {\sum_{i=1}^{\nu}F_{i}^{\sfc}{\left( {\rm Im}(\hr) \right)}}
     - |\whX_{\sfd}|^2, \label{eq:mseX_RealGaussian}
\end{align}
where
\begin{align*}
    F_{i}^{\sfs}(x) =  e^{-\frac{(2i-1)^2\zeta}{v^{r}}} \sinh{\left( \frac{2(2i-1)\zeta}{v^{r}} x \right)}, \\
    F_{i}^{\sfc}(x) =  e^{-\frac{(2i-1)^2\zeta}{v^{r}}} \cosh{\left( \frac{2(2i-1)\zeta}{v^{r}} x \right)}.
\end{align*}
Finally, recall that ${\sfP_{\sfH}{(H_{nk})} = \calN_{\bbC}(0,\sigma_{h}^2)}$. Then lines 17--18 of Algorithm \ref{ago:BiGAMP-JCD} can be computed using
\begin{equation}
    \whH =  \frac{\sigma_{h}^2}{\sigma_{h}^2+v^{q}} \, \hq ~~\mbox{and}~~
    v^{h} = v^{q} - \frac{ (v^{q})^2}{\sigma_{h}^2+v^{q}}.\label{eq:hatH_RealGaussian}
\end{equation}
The derivation of (\ref{eq:hatH_RealGaussian}) is identical to that in Example 2.

Using the above nonlinear steps (\ref{eq:hatZ_RealGaussian})--(\ref{eq:mseZ_RealGaussian}) and (\ref{eq:hatX_RealGaussian})--(\ref{eq:hatH_RealGaussian}), we
implement the GAMP-based JCD algorithm based on the open-source ``GAMPmatlab'' software suite.

\section*{ IV. Performance Analysis}
In this section, we present a framework to analyze the Bayes-optimal JCD estimator. The key strategy for analyzing $\mse(\qH)$ and $\mse(\qX_{\sfd})$ is through the average free entropy
\begin{equation}\label{eq:FreeEn}
    \calF \triangleq \frac{1}{K^2}\Ex_{\wtqY}{\left\{\log \sfP(\wtqY) \right\}},
\end{equation}
where $\sfP(\wtqY)$ denotes the marginal likelihood in (\ref{eq:marginalPr}), that is, the partition function. Aligned with the argument in \cite{Krzakala-13ISIT,Kabashima-14ArXiv}, $\mse(\qX_{\sfd})$ and $\mse(\qH)$ are saddle points of the average free entropy. Thus, the goal is reduced to finding (\ref{eq:FreeEn}).

The analysis is based on a large-system limit, that is, when $N, K, T \to \infty$ but the ratios
\begin{equation}
N/K= \alpha,~~T/K = \beta,~~ T_\sft/K =\beta_\sft,~~ T_\sfd/K =\beta_\sfd,
\end{equation}
are fixed and finite. For convenience, we simply use $K \rightarrow \infty$ to denote this large-system limit. Even in the large-system limit, the computation of (\ref{eq:FreeEn}) is difficult. The major difficulty in computing (\ref{eq:FreeEn}) is the expectation of the logarithm of $\sfP(\wtqY)$, which, nevertheless, can be
facilitated by rewriting $\calF$ as\footnote{We use the following formula from right to left
\begin{equation*}
    \lim_{\tau \rightarrow 0} \frac{ \partial }{ \partial \tau}  \log \Ex\{ {\sf A}^\tau \} =
    \lim_{\tau \rightarrow 0}  \frac{\Ex\{  {\sf A}^\tau \log {\sf A} \}}{ \Ex\{ {\sf A}^\tau \} }
    = \Ex\{ \log {\sf A} \},
\end{equation*}
where ${\sf A}$ is any positive random variable.}
\begin{equation}\label{eq:LimF}
\calF = \frac{1}{K^2} \lim_{\tau\rightarrow 0}\frac{\partial}{\partial \tau}\log \Ex_{\wtqY}{\left\{\sfP^{\tau}(\wtqY)\right\}}.
\end{equation}

The expectation operator is then transformed inside the log-function. We first evaluate $\Ex_{\wtqY}\{\sfP^{\tau}(\wtqY)\}$ for an integer-valued $\tau$, and then generalize the result to any positive real number $\tau$. This technique, called {\em the replica method}, is from the field of statistical physics \cite{Nishimori-01BOOK}, which is not mathematically rigorous. Nevertheless, the replica method has proved successful in a number of highly difficult problems in statistical physics \cite{Nishimori-01BOOK} and information theory \cite{Tanaka-02IT,Moustakas-03TIT,Guo-05IT,Muller-03TSP,Wen-07IT,Hatabu-09PRE,Takeuchi-13TIT,Girnyk-14TWC}.
Some of the results originally obtained by the replica method have been subsequently validated by other approaches (e.g., \cite{Bayati-11IT,Zhang-13JSAC}).
Under the assumption of $K \rightarrow \infty$ and replica symmetry (RS), an asymptotic free entropy can be obtained later in Proposition \ref{Pro_FreeEntropy}. We check the accuracy of the replica-based analysis via simulations.
Proposition \ref{Pro_FreeEntropy} involves several new parameters.

\subsection*{A. Parameters of Proposition \ref{Pro_FreeEntropy}}
Most parameters (except for some auxiliary parameters) of Proposition \ref{Pro_FreeEntropy} can be illustrated systematically by the scalar AWGN channels:
\begin{subequations} \label{eq:scalCh}
\begin{align}
 Y_{X_{\sfd}} &=  \sqrt{\tq_{{X_{\sfd}}}} X_{\sfd} + W_{X_{\sfd}}, \label{eq:scalCh_X} \\
 Y_{H} &= \sqrt{\tq_{H}} H + W_{H}, \label{eq:scalCh_H}
\end{align}
\end{subequations}
where $W_{H},W_{X_{\sfd}} \sim \calN_{\bbC}(0,1)$, $H \sim \sfP_{H}$, and $X_{\sfd} \sim  \sfP_{X_{\sfd}}$. We shall specify how the parameters $\tq_{H}$ and $\tq_{{X_{\sfd}}}$ are related to the asymptotic free entropy later in Proposition \ref{Pro_FreeEntropy}. Thus far, we know that the parameters $\tq_{{X_{\sfd}}}$ and $\tq_{H}$ serve as the signal-to-noise ratios (SNRs) of the above AWGN channels. The likelihoods under (\ref{eq:scalCh_X}) and (\ref{eq:scalCh_H}) are  respectively given by
\begin{subequations}
\begin{align}
 \sfP(Y_{{X_{\sfd}}}|X_{\sfd}) &= \frac{1}{\pi}e^{-|Y_{{X_{\sfd}}}-\sqrt{\tq_{{X_{\sfd}}}} {X_{\sfd}}|^2}, \\
 \sfP(Y_{H}|H) &= \frac{1}{\pi}e^{-|Y_{H}-\sqrt{\tq_{H}}H|^2},
\end{align}
\end{subequations}
and then we get the posteriors
\begin{subequations}
\begin{align}
 \sfP(X_{\sfd}|Y_{{X_{\sfd}}}) &=
 \frac{{\sfP_{\sfX_{\sfd}}( X_{\sfd} )} \sfP(Y_{{X_{\sfd}}}|X_{\sfd}) }
    { \int\!\rmd x'_{\sfd}\, {\sfP_{\sfX_\sfd}( x'_{\sfd} )}\sfP(Y_{{X_{\sfd}}}|x'_{\sfd}) }, \\
 \sfP(H|Y_{H}) &= \frac{{\sfP_{\sfH}( H )} \sfP(Y_{{H}}|H) }
    { \int\!\rmd h' \, {\sfP_{\sfH}( h' )}\sfP(Y_{{H}}|h') }.
\end{align}
\end{subequations}

Some important quantities are obtained with the posteriors. For example, the posterior mean estimators for $X_{\sfd}$ and $H$ read
\begin{subequations} \label{eq:est_XH}
\begin{align}
 \whX_{\sfd} &= \int\!\rmd X_{\sfd} \, X_{\sfd}\sfP(X_{\sfd}|Y_{{X_{\sfd}}}), \\
 \whH &= \int\!\rmd H \, H \sfP(H|Y_{H}).
\end{align}
\end{subequations}
The MSEs of the estimators are thus given by
\begin{subequations} \label{eq:mse_XH}
\begin{align}
 \mse_{X_{\sfd}} &= \Ex\left\{|X_{\sfd}- \whX_{\sfd}|^2\right\}, \\
  \mse_{H} &= \Ex\left\{|H - \whH |^2\right\},
\end{align}
\end{subequations}
where the expectations are taken over ${\sfP( Y_{X_{\sfd}},X_{\sfd} )} $ and ${\sfP( Y_{H},H )}$, respectively. Equations (\ref{eq:hatX_RealGaussian})--(\ref{eq:hatH_RealGaussian}) are explicit expressions of the above quantities. In addition, the mutual information between $Y_{{X_{\sfd}}}$ and ${X_{\sfd}}$ reads \cite{Cover-BOOK}
\begin{equation} \label{eq:MIX_scalCh}
  I({X_{\sfd}};Y_{{X_{\sfd}}}|\tq_{{X_{\sfd}}})
   = -\Ex_{Y_{{X_{\sfd}}}} {\left\{ \log \Ex_{{{X_{\sfd}}}}\left\{ e^{-\left|Y_{{X_{\sfd}}} - \sqrt{\tq_{{X_{\sfd}}}} {X_{\sfd}} \right|^2} \right\} \right\}} - 1,
\end{equation}
and the mutual information between $Y_{H}$ and ${H}$ is
\begin{equation} \label{eq:MIH_scalCh}
  I({H};Y_{H}|\tq_{H}) = -\Ex_{Y_{H}} {\left\{ \log \Ex_{{H}}\left\{ e^{-\left|Y_{H} - \sqrt{\tq_{H}}{H} \right|^2} \right\} \right\}}- 1.
\end{equation}

From (\ref{eq:scalCh}), an inference that another scalar AWGN channel w.r.t.~the $\sft$-phase exists can be made, i.e.,
\begin{equation}
 Y_{X_{\sft}} =  \sqrt{\tq_{{X_{\sft}}}} X_{\sft} + W_{X_{\sft}}, \label{eq:scalCh_Xt}
\end{equation}
where ${W_{X_{\sft}} \sim \calN_{\bbC}(0,1)}$ and ${X_{\sft} \sim  \sfP_{X_{\sft}}}$. As the pilot is known, ${\mse_{X_{\sft}} = 0}$ can be obtained easily following the argument in Remark \ref{Remark_1}; and the mutual information between $Y_{{X_{\sft}}}$ and ${X_{\sft}}$ is $0$. As all the performances relating to (\ref{eq:scalCh_Xt}) are trivial, we will not use (\ref{eq:scalCh_Xt}) in the following discussions.

\subsection*{B. Analytical Results}

\begin{Proposition} \label{Pro_FreeEntropy}
As $K \to \infty$, the asymptotic free entropy is
\begin{align}
 \calF &= \alpha \sum_{o \in \{\sft,\,\sfd\} } \beta_{o} \Bigg( \sum_{b=1}^{2^{\sfB}} \int\!\rmD v\,
   \Psi_{b}\left( V_{o} \right)
   \log \Psi_{b}\left( V_{o} \right)
    \Bigg) \notag \\
  & \hspace{0.35cm} - \alpha  I(H;Y_{H}|\tq_{H})
  - \beta_{\sfd}  I(X_{\sfd};Y_{X_{\sfd}}|\tq_{X_{\sfd}}) \notag \\
  & \hspace{0.35cm} + \alpha (c_{H} - q_{H})\tq_{H} + \sum_{o \in \{\sft,\,\sfd\} } \beta_{o} (c_{X_{o}}-q_{X_{o}}) \tq_{X_{o}}, \label{eq:FreeEntropyFinal}
\end{align}
where
\begin{multline}\label{eq:Psi_def}
 \Psi_{b}(V_{o})
\triangleq \Phi{\left(\frac{ \sqrt{2}r_{b}-V_{o}}{\sqrt{\sigma_{w}^2 + c_{H}c_{X_{o}}-q_{H}q_{X_{o}}}} \, \right)}\\
\quad - \Phi{\left(\frac{ \sqrt{2}r_{b-1}-V_{o}}{\sqrt{\sigma_{w}^2 + c_{H}c_{X_{o}}-q_{H}q_{X_{o}}}} \, \right)};
\end{multline}
${V_{o} \triangleq \sqrt{q_{H}q_{X_{o}}} v}$ for $o \in \{\sft,\,\sfd\}$;
$I(\cdot)$'s are given by (\ref{eq:MIX_scalCh}) and (\ref{eq:MIH_scalCh}); and
$c_{X_{o}} \triangleq \Ex\{|X_{o}|^2\} = \sigma_{x_{o}}^2$, $c_{H} \triangleq \Ex\{|H|^2\} = \sigma_{h}^2$. In (\ref{eq:FreeEntropyFinal}), the other parameters $\{ q_{X_{o}}, q_{H}, \tq_{X_{o}},\tq_{H} \}$ are obtained from the solutions to the fixed-point equations
\begin{subequations} \label{eq:fxiedPoints}
\begin{align}
 \tq_{H} & = \beta_{\sft} q_{X_{\sft}} \chi_{\sft} + \beta_{\sfd} q_{X_{\sfd}} \chi_{\sfd},
 && q_{H} = c_{H} - \mse_{H}, \label{eq:asyVarH} \\
 \tq_{X_{\sft}} & = \alpha q_{H} \chi_{\sft},
 && q_{X_{\sft}} = c_{X_{\sft}} - \mse_{X_{\sft}}, \label{eq:asyVarXt} \\
 \tq_{X_{\sfd}} &= \alpha q_{H} \chi_{\sfd},
 && q_{X_{\sfd}}  = c_{X_{\sfd}} - \mse_{X_{\sfd}}, \label{eq:asyVarXd}
\end{align}
\end{subequations}
where $\mse_{X_{\sft}} = 0$, and $\mse_{H}$ and $\mse_{X_{\sfd}}$ are given by (\ref{eq:mse_XH}). In (\ref{eq:fxiedPoints}), we have defined
\begin{equation} \label{eq:chi_def}
 \chi_{o} \triangleq \sum_{b=1}^{2^{\sfB}} \int\!\rmD v \frac{\Big(\Psi'_{b}\left(\sqrt{q_{H}q_{X_{o}}} v \right)\Big)^2}{\Psi_{b}\left(
 \sqrt{q_{H}q_{X_{o}}} v \right)},~\mbox{for }{o \in \{\sft,\,\sfd\}}
\end{equation}
with $\Psi_{b}(\cdot)$ given by (\ref{eq:Psi_def}) and
\begin{multline}  \label{eq:Psip_def}
\Psi'_{b}(V_{o})\triangleq \frac{\partial \Psi_{b}(V_{o})}{\partial V_{o}}\\
=  \frac{ e^{-\frac{(\sqrt{2}r_{b}-V_{o})^2}{2(\sigma_{w}^2 + c_{H}c_{X_{o}}-q_{H}q_{X_{o}})}} -e^{-\frac{(\sqrt{2}r_{b-1}-V_{o})^2}{2(\sigma_{w}^2 + c_{H}c_{X_{o}}-q_{H}q_{X_{o}})}}}{\sqrt{2 \pi (\sigma_{w}^2 + c_{H}c_{X_{o}}-q_{H}q_{X_{o}})}}.
\end{multline}
\end{Proposition}
\begin{proof}
See Appendix B.
\end{proof}

As previously mentioned, the asymptotic MSEs of $\qX_{\sfd}$ and $\qH$ are the saddle points of the free entropy. Clearly, from Proposition \ref{Pro_FreeEntropy},
they are $\mse_{X_{\sfd}}$ and $\mse_{H}$, respectively. Notably, the MSEs are associated with the scalar AWGN channels (\ref{eq:scalCh_X}) and (\ref{eq:scalCh_H}). Therefore, the performances of the quantized MIMO system seem to be fully characterized by the scalar AWGN channels (\ref{eq:scalCh}). The following proposition formulates such intuition.

\begin{Proposition} \label{Pro_Decoupling}
Let $X_{\sfd,kt}$, $H_{nk}$, $\whX_{\sfd,kt}$, and $\whH_{nk}$ denote the ${(k,t)}$-th and ${(n,k)}$-th entries of $\qX_{\sfd}$, $\qH$, $\whqX_{\sfd}$, and $\whqH$, respectively. As $K \to \infty$, the joint distribution of $(X_{\sfd,kt},H_{nk},\whX_{\sfd,kt},\whH_{nk})$ of channels (\ref{eq:qsys}), (\ref{eq:estH}), and (\ref{eq:estX}) converges to the joint distribution of $(X_{\sfd},H,\whX_{\sfd},\whH)$ of the scalar channels (\ref{eq:scalCh_X}) and (\ref{eq:scalCh_X}).
\end{Proposition}

\begin{proof}
See Appendix C.
\end{proof}

Proposition \ref{Pro_Decoupling} shows that, in the large-system limit, the input output of the quantized MIMO system employing the Bayes-optimal JCD estimator is equivalently decoupled into a bank of the scalar AWGN channels (\ref{eq:scalCh_X}) and (\ref{eq:scalCh_H}). This characteristic is known as the decoupling principle, which was introduced by \cite{Guo-05IT} for treading an \emph{unquantized} MIMO system with \emph{perfect} CSIR. If perfect CSIR is available, then we will not need (\ref{eq:scalCh_H}) for treating the channel estimation quality. Notably, Proposition \ref{Pro_Decoupling} extends the decoupling principle to a general setting. In particular, we employ the JCD estimator so that the decoupled AWGN channels involve not only the data symbol [i.e, (\ref{eq:scalCh_X})] but also the channel response [i.e., (\ref{eq:scalCh_H})] as well.

\begin{Remark}
The equivalent channels (\ref{eq:scalCh}) are the scalar AWGN channels, with $\tq_{\sfd}$ and $\tq_{H}$ being the equivalent SNRs. As shown in (\ref{eq:fxiedPoints}) and (\ref{eq:chi_def}), the quantization effect is included in $\tq_{\sfd}$ and $\tq_{H}$ through $\chi_{o}$ for $o \in \{ \sft,\,\sfd\}$. Consider the extreme case of $\sfB \rightarrow \infty$ and $\Delta \rightarrow 0$, i.e., the unquantized channel. In this case, the Riemann sums $\sum_{b=1}^{2^{\sfB}}$ in (\ref{eq:chi_def}) becomes the Riemann integral over the interval $(-\infty, \infty)$. Applying the technique in {\rm Remark}~3 to (\ref{eq:chi_def}) and evaluating the integrals, $\chi_{o}$ can be simplified to
\begin{equation} \label{eq:chi_UnQ}
    \chi_{o} = \frac{1}{ \sigma_{w}^2 + c_{H}c_{X_{o}}-q_{H}q_{X_{o}} }.
\end{equation}
Substituting (\ref{eq:fxiedPoints}) for $q_{H}$ and $q_{X_{o}}$ in the denominator of (\ref{eq:chi_UnQ}), we obtain $\sigma_{w}^2 + c_{H}\mse_{X_{o}}+(c_{X_{o}} - \mse_{X_{o}})\mse_{H}$. The quantity in this form can be understood as the noise plus the residual interference resulting from the estimation errors of the data symbol and channel response.
\end{Remark}

We focus our interest on several special cases in the following examples to obtain more insight into Proposition
\ref{Pro_Decoupling}.

{\noindent {\bf Example~3} (Constellation-like Inputs).} Based on Proposition \ref{Pro_Decoupling}, the asymptotic MSEs w.r.t.~$\qX_{\sfd}$ and $\qH$ can be determined by using the MSEs of the scalar AWGN channels (\ref{eq:scalCh_X}) and (\ref{eq:scalCh_H}), respectively. Thus, if the data symbol is drawn from a quadrature phase-shift keying (QPSK) constellation, then we will derive
\begin{align}
    \mse_{X_{\sfd}} &= 1-\int{\rmD z \tanh{\left(\tq_{X_{\sfd}}+\sqrt{\tq_{X_{\sfd}}}z\right)}},
    \label{eq:mseQPSK_X} \\
    \mse_{H} &= \frac{\sigma_{h}^2}{1+\sigma_{h}^2\tq_{H}}. \label{eq:mseQPSK_H}
\end{align}
The SER w.r.t.~$\qX_{\sfd}$ can also be evaluated through the scalar AWGN channel (\ref{eq:scalCh_X}), which is given by \cite[p.269]{Proakis-95}
\begin{equation} \label{eq:SER}
    {\sf SER} =  2{\cal Q}{\left(\sqrt{\tq_{X_{\sfd}}}\right)} - \left[{\cal Q}{\left(\sqrt{\tq_{X_{\sfd}}}\right)}\right]^2,
\end{equation}
where ${\cal Q}(x) = \int_{x}^{\infty}\rmD z$ is the Q-function.

In fact, all these performances w.r.t.~$\qX_{\sfd}$ can be determined on the basis of knowledge of the scalar AWGN channel with SNR $\tq_X$. Thus, if the data symbol is drawn from other square QAM constellations, then the corresponding SER can be easily obtained by using the closed-form SER expression in \cite[p.279]{Proakis-95}. 

{\noindent {\bf Example~4} (Perfect CSIR).} If the channel matrix $\qH$ is perfectly known, then the $\sft$-phase will not be required so that
\begin{equation} \label{eq:beta_PCSI}
 \beta_{\sft} = 0 ~~~\mbox{and}~~~ \beta_{\sfd} = \beta.
\end{equation}
Given that $\qH$ is perfectly known, $\mse_{H} =0$. Integrating this into (\ref{eq:asyVarH}), we immediately obtain ${q_{H} = c_{H} = \sigma_{h}^2}$, which yields
\begin{align}
    q_{H}q_{X_{\sfd}} &= c_{H} q_{X_{\sfd}}, \label{eq:qq_PCSI} \\
    {c_{H}c_{X_{\sfd}}-q_{H}q_{X_{\sfd}}} &= {c_{H} \mse_{X_{\sfd}}}, \label{eq:ccpqq_PCSI}
\end{align}
in which (\ref{eq:ccpqq_PCSI}) follows the result of ${c_{H}c_{X_{\sfd}}-q_{H}q_{X_{\sfd}}}  = {c_{H} (c_{X_{\sfd}}-q_{X_{\sfd}})}$ and (\ref{eq:asyVarXd}). Substituting (\ref{eq:beta_PCSI})--(\ref{eq:ccpqq_PCSI}) into (\ref{eq:Psi_def}), (\ref{eq:chi_def}) and (\ref{eq:Psip_def}), we derive more concise expressions for $\chi_{\sfd}$, $\Psi_{\sfb}(\cdot)$, and $\Psi'_{\sfb}(\cdot)$. Notably, when particularizing our results to the case with the QPSK inputs, we recover the same asymptotic MSE expression as given in \cite[(7) and (8)]{Nakamura-08ISITA}. More precisely, in \cite{Nakamura-08ISITA}, the real-valued system with BPSK signal was considered. In such case, $\sqrt{2}r_{b}$ in our study should be replaced with $r_{b}$.

{\noindent {\bf Example~5} (Pilot-only Scheme).} In the conventional pilot-only scheme, the receiver solely uses $\wtqY_{\sft}$ and $\qX_{\sft}$ to generate an estimate of $\qH$ and subsequently uses the estimated channel for estimating the data $\qX_{\sfd}$ from $\wtqY_{\sfd}$ \cite{Risi-14ArXiv}. The analysis of the asymptotic MSE w.r.t.~$\qH$ is the same as that in Example~4, but the roles of $\qH$ and $\qX_{\sft}$ are exchanged. In particular, during the $\sft$-phase, we derive $\beta_{\sfd} = 0$ and $\mse_{X_{\sft}} = 0$ because no data symbol is involved and the pilot matrix $\qX_{\sft}$ is known. After substituting these parameters into (\ref{eq:fxiedPoints}) and simplification, we obtain the \emph{self-contained} fixed-point equations
\begin{align}
    &\mse_{H} = \frac{\sigma_{h}^2}{1+\sigma_{h}^2\tq_{H}}, \label{eq:mse_H_PilotOnly}\\
    &\tq_{H}  = \beta_{\sft} \sigma_{x_{\sft}}^2 \chi_{\sft} \label{eq:q_H_PilotOnly}
\end{align}
with
\begin{equation} \label{eq:chi_PilotOnly}
 \chi_{\sft} =
 \sum_{b=1}^{2^{\sfB}} \int\!\rmD v \frac{\Big(\Psi'_{b}\left(\sqrt{\sigma_{x_{\sft}}^2(\sigma_{h}^2-\mse_{H})} v \right)\Big)^2}{\Psi_{b}\left(
 \sqrt{\sigma_{x_{\sft}}^2(\sigma_{h}^2-\mse_{H}) } v \right)}.
\end{equation}
In (\ref{eq:mse_H_PilotOnly}), $\mse_{H}$ represents the asymptotic MSE w.r.t.~$\qH$ for the pilot-only scheme, which is also the MSE w.r.t. $H$ for the scalar AWGN channel (\ref{eq:scalCh_H}). Notably, $\tq_{H}$ serves as the SNR of the AWGN channel (\ref{eq:scalCh_H}). Comparing $\tq_{H}$ in (\ref{eq:q_H_PilotOnly}) with that in (\ref{eq:asyVarH}), we determine that the second term of $\tq_{H}$ in (\ref{eq:asyVarH}) is the gain achieved by data-aided channel estimation.

\begin{table*}[t]
\caption{$C_{\sfB}$ for uniform $\sfB$-bit quantizer with $\Delta = \sqrt{\sfB}2^{-\sfB}$. } 
\centering 
\begin{tabular}{c|ccccccc}
\hline
  $\sfB$ & $1$ & $2$ & $3$ & $4$ & $5$ & $6$ & $7$\\
\hline \hline
  $C_{\sfB}$ (in dB) & $2.8731$ & $-5.9852$ & $-13.0201$ & $-19.4804$ & $-25.7065$ & $-31.8265$ & $-37.6547$ \\
\hline
\end{tabular}
\label{table:Coefficient_CB}
\end{table*}

\begin{figure}
\begin{center}
\resizebox{3in}{!}{%
\includegraphics*{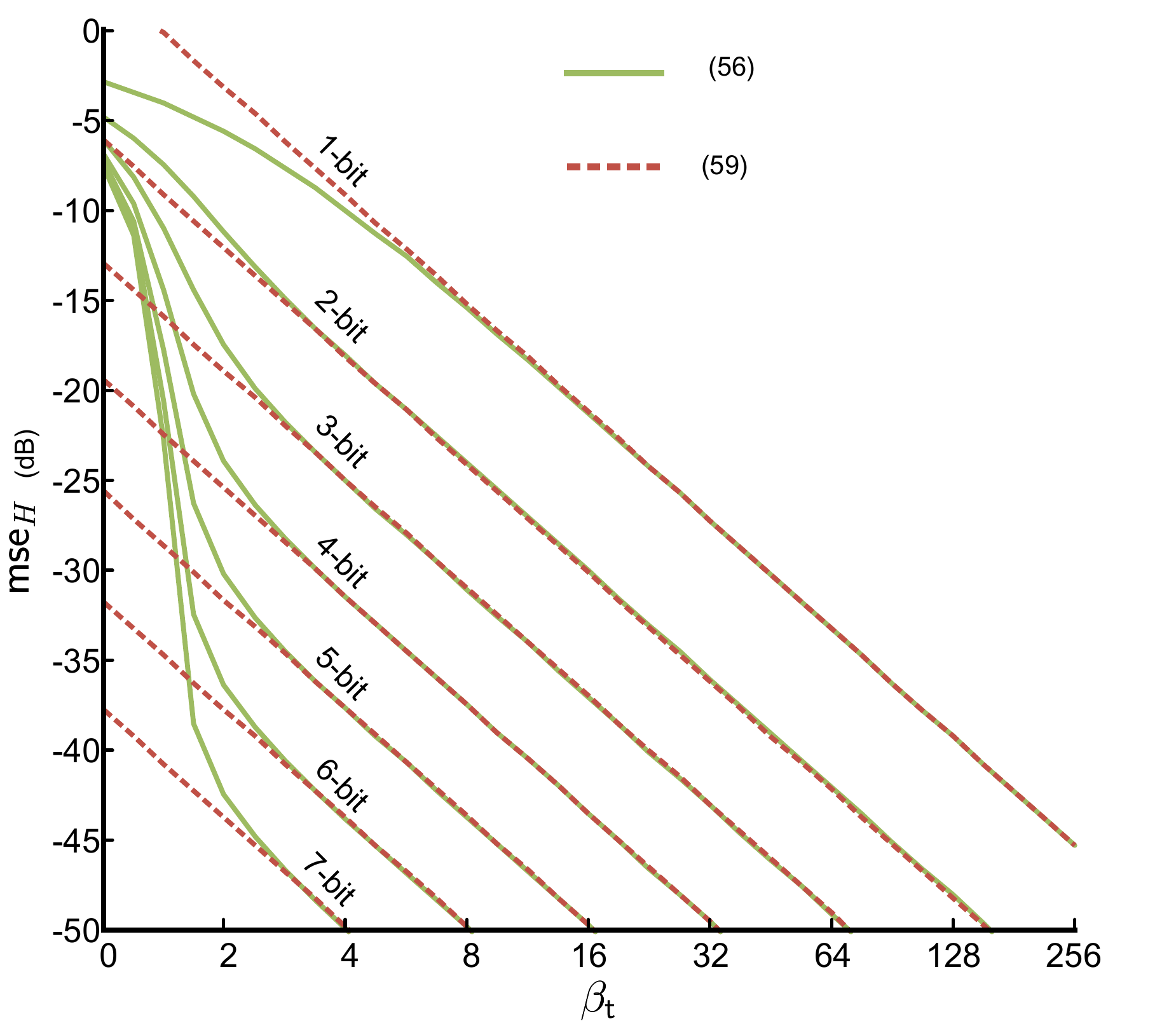} }%
\caption{The asymptotic MSE w.r.t.~$\qH$ of the pilot-only scheme versus the pilot ratio $\beta_{\sft} = T_{\sft}/K$ for different $\sfB$-bit quantizer. 
}\label{fig:MSEH_PilotOnly}
\end{center}
\end{figure}

Before proceeding with the analysis of estimated data in the pilot-only scheme, we provide the following proposition to obtain a better understanding of $\mse_{H}$ in (\ref{eq:mse_H_PilotOnly}).

\begin{Proposition} \label{pro:mseOfH_PilotOnly}
Let the channel gain and the transmit pilot power be normalized, that is, ${\sigma_{h}^2 = 1}$ and ${\sigma_{x_{\sft}}^2 = 1}$. In the high-SNR regime and $\beta_{\sft} = T_{\sft}/K \gg 1$, $\mse_{H}$ of the pilot-only scheme can be approximately expressed as
\begin{equation} \label{eq:mse_HighSNR}
 \mse_{H} \approx {-20\log_{10}(\beta_{\sft}) + C_{\sfB}}~({\rm dB}),
\end{equation}
where $C_{\sfB}$ is a quantizer-dependent (e.g., $\Delta$ and $\sfB$) constant.
\end{Proposition}

\begin{proof}
See Appendix D.
\end{proof}

As an example, Table \ref{table:Coefficient_CB} provides the corresponding value of $C_{\sfB}$ for a uniform $\sfB$-bit quantizer with ${\Delta = \sqrt{\sfB}2^{-\sfB}}$. In this case, we plot the MSEs that use the approximate expression (\ref{eq:mse_HighSNR}) as well as its analytical form (\ref{eq:mse_H_PilotOnly}) in Figure~\ref{fig:MSEH_PilotOnly}. We observe that, for ${\beta_{\sft} > 2}$, the approximation (\ref{eq:mse_HighSNR}) matches the theoretical result (\ref{eq:mse_H_PilotOnly}) perfectly. Table \ref{table:Coefficient_CB} shows that the constant $C_{\sfB}$ satisfies ${C_{\sfB} \approx -6.02\sfB + 4.4895}$ in high-resolution cases, indicating that $\mse_{H}$ decreases by $6$\,dB for each 1 bit increase of the rate. Notably, this property coincides with the well-known figure of merit in quantization.\footnote{The property of $6$\,dB improvement in signal-to-quantization-noise ratio for each extra bit is a well-known figure of merit in the ADC literature \cite[p.248]{Oppenheim-09BOOK}.} From (\ref{eq:mse_HighSNR}), given a fixed quantizer (i.e., fixed $C_{\sfB}$), $\mse_{H}$ increases by $6$\,dB for each doubling of training length $\beta_{\sft}$. Consequently, doubling the length for training has the same effect as increasing an extra bit on every ADC at the massive MIMO receiver.

The proceeding observation provides a useful guideline for the trade-off between the training length and the ADC word length. For instance, if we target
$\beta_{\sft}$ to that attained by the pilot-only scheme at $\mse_{H} = -30$dB, the 4 bit receiver requires $\beta_{\sft} = 4$, as shown in
Figure~\ref{fig:MSEH_PilotOnly}. If we intend to reduce the ADC word length of each ADC to $1$ bit, then the training length increases $2^{4-1} = 8$ times compared
with that in the 4 bit case. This argument shows the importance of the JCD technique in the \emph{quantized} MIMO system. With the JCD technique, the estimated
payload data are utilized to aid channel estimation so that the \emph{effective} training length virtually increases.

Then, we return to the analysis of estimated data. If the channel estimate is subsequently used to estimate data via the Bayes-optimal approach, then we can derive the corresponding \emph{self-contained} fixed-point equations for the $\sfd$-phase similar to (\ref{eq:mse_H_PilotOnly}) and (\ref{eq:q_H_PilotOnly}). In particular, we calculate (\ref{eq:asyVarXd}) given a \emph{fixed} ${q_{H} = \sigma_{h}^2 - \mse_{H}}$, with $\mse_{H}$ given in (\ref{eq:mse_H_PilotOnly}). Given that no iteration process occurred between the pilots and data symbols, (\ref{eq:asyVarH}) and (\ref{eq:asyVarXt}) are not involved in the $\sfd$-phase. If the JCD technique is employed, then $\mse_{H}$ can be further reduced. Any reduction in channel estimation error $\mse_{H}$ results in an increase in $q_{H}$; thus, $\tq_{X_{\sfd}} = \alpha q_{H} \chi_{\sfd}$ increases.

\section*{ V. Discussions and Numerical Results}

\begin{figure}
\begin{center}
\resizebox{3.75in}{!}{%
\includegraphics*{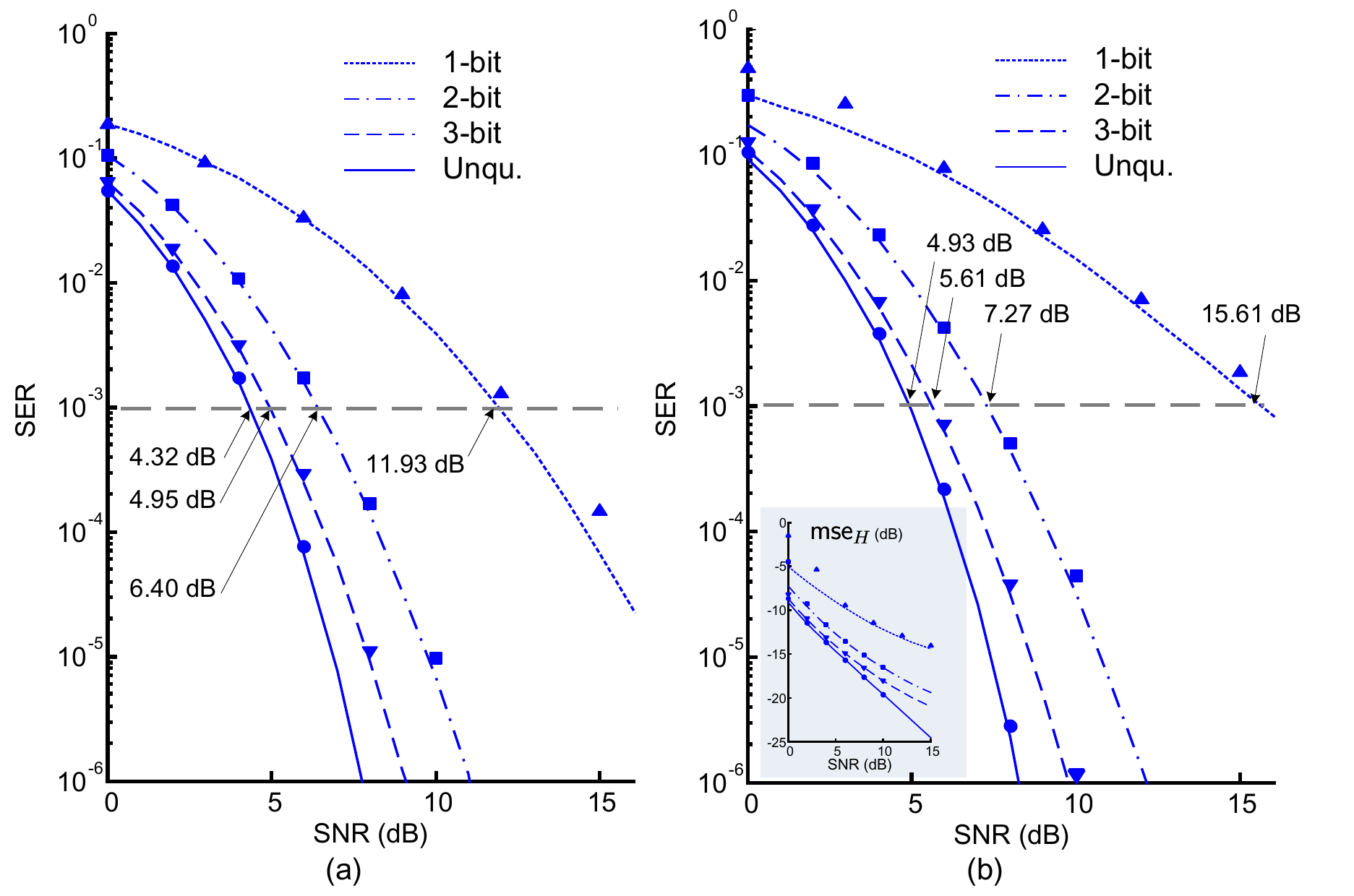} }%
\caption{SER versus SNR for QPSK constellations. In the results, the JCD estimation scheme is used under the settings with a) perfect CSIR and b) no CSIR. Curves denote analytical results and markers denote Monte-Carlo simulation results achieved by the GAMP-based JCD algorithm. The MSEs  w.r.t.~$\qH$ of the JCD estimator are plotted as a subfigure.}\label{fig:SER_QPSK}
\end{center}
\end{figure}

\subsection*{A. Accuracy of the Analytical Results}
Computer simulations are conducted to verify the accuracy of our analytical results. In particular, we compare the SER expression (\ref{eq:SER}) and the analytical
MSE w.r.t.~$\qH$ (\ref{eq:mseQPSK_H}) with those obtained by the simulations. The simulation results are obtained by averaging over $10,000$ channel realizations,
wherein the GAMP-based JCD algorithm (Algorithm 1) is implemented with tolerance $\epsilon  = 10^{-8}$ and the maximum number of iterations
$\xi_{\max}=100$.\footnote{ We do not show the convergence of the GAMP-based JCD algorithm because of space limitation. In most cases, the GAMP-based JCD algorithm
converges after only $20$--$30$ iterations although it shows a slow convergence at low SNRs.} The system parameters are set as follows: ${K=50}$, ${N=200}$,
${T_\sft = 50}$, and ${T_\sfd = 450}$. The SNR of the system is defined as ${\rm SNR}=1/\sigma_{w}^2$. The pilot matrix $\qX_{\sft} \in \bbC^{K \times T_{\sft}}$
consists of statistically independent QPSK constellations. In the simulations, we use the typical uniform quantizer with a fixed quantization step size ${\Delta =
1/2}$. Notably, this quantization step size is nonoptimal. The optimal step size will be discussed in the subsequent subsection. As the QPSK constellations are
used to transmit data, Figure~\ref{fig:SER_QPSK} shows the corresponding SER results for the cases of (a) perfect CSIR and (b) no CSIR. The corresponding MSE
w.r.t.~$\qH$ of the JCD estimator are plotted as a subfigure in Figure~\ref{fig:SER_QPSK}(b). We observe that the GAMP-based JCD algorithm can generally describe
the performances of the \emph{theoretical} Bayes-optimal estimator. The performance of the theoretical Bayes-optimal estimator can also be described by our
analytical expressions. Notably, the GAMP-based JCD algorithm is only an approximation of the Bayes-optimal JCD estimator whose implementation is prohibitive. For
the case with no CSIR, the GAMP-based JCD algorithm cannot work as well as that predicted by the analytical result at low SNRs. At low SNRs, the GAMP-based JCD
algorithm shows a slow convergence, such that the adopted maximum number of iterations is insufficient.\footnote{At low SNRs, we obtain a good result by increasing
the maximum number of iterations.} This gap motivated the search for other improved estimators in the future.

Figure~\ref{fig:SER_QPSK}(b) shows that performance degradation due to low-precision quantization is small. For instance, if we target the SNR to that attained by the unquantized system at SER$=10^{-3}$, then the $3$ bit Bayes-optimal JCD estimator only incurs a loss of ${5.61-4.32=1.29}$\,dB. Even with $2$-bit quantization, the loss of ${7.27-4.32=2.95}$\,dB remains acceptable.

\subsection*{B. Optimal Step-Size}

\begin{figure}
\begin{center}
\resizebox{2.8in}{!}{%
\includegraphics*{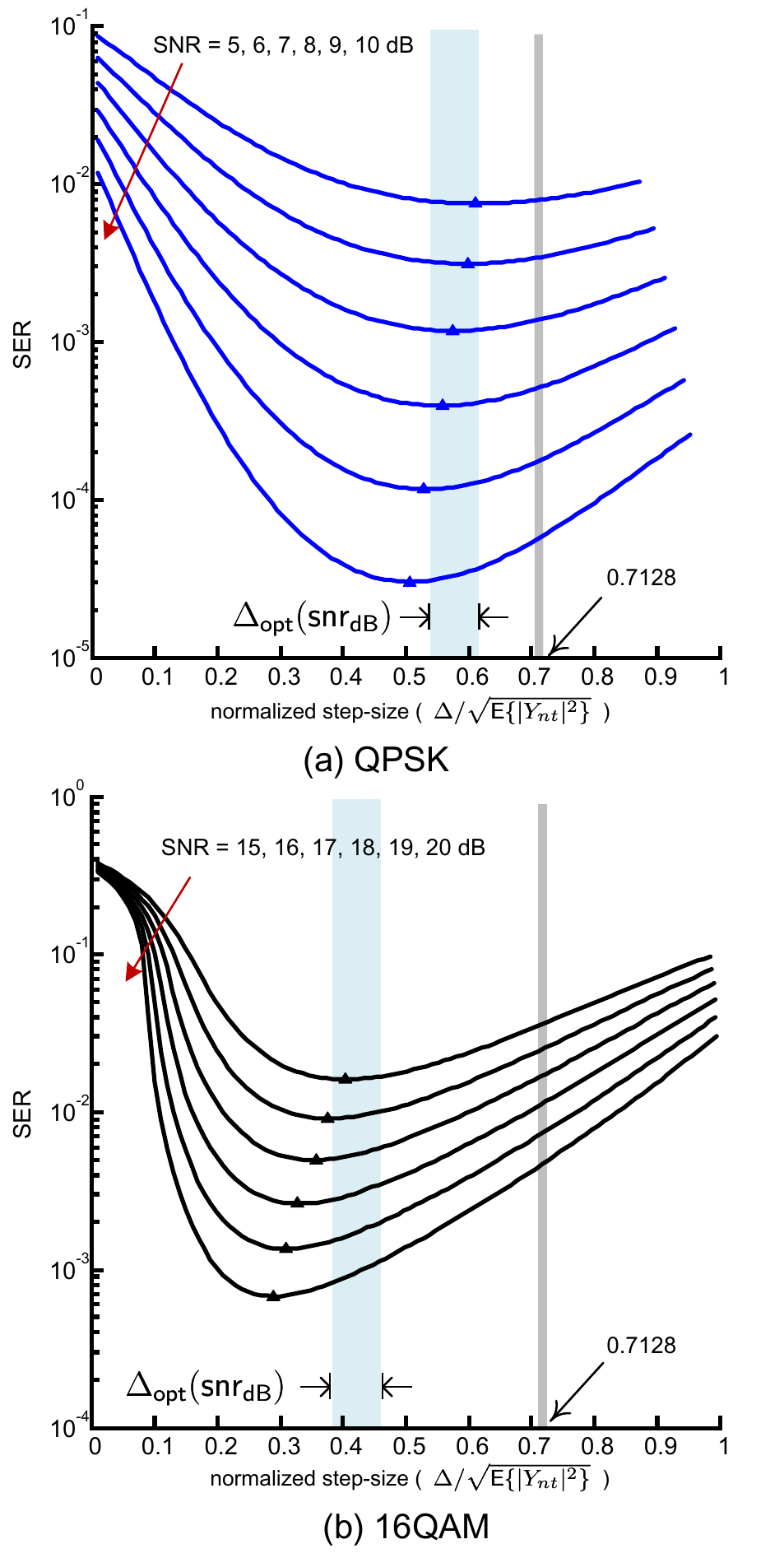} }%
\caption{SERs versus the normalized step size under the quantized MIMO system with a) QPSK and b) 16QAM constellations for $\alpha = 4$, $\beta = 10$, $\beta_{\sft} = 1$. Markers correspond to the lowest SER w.r.t.~the normalized step size. The optimal step size determined by minimizing the distortion of a Gaussian signal \cite{Dhahir-TSP96}, i.e., $\Delta= 0.7128$, is plotted as the vertical axis.}\label{fig:StepSize_QPSK16QAM_SER}
\end{center}
\end{figure}

\begin{figure}
\begin{center}
\resizebox{2.8in}{!}{%
\includegraphics*{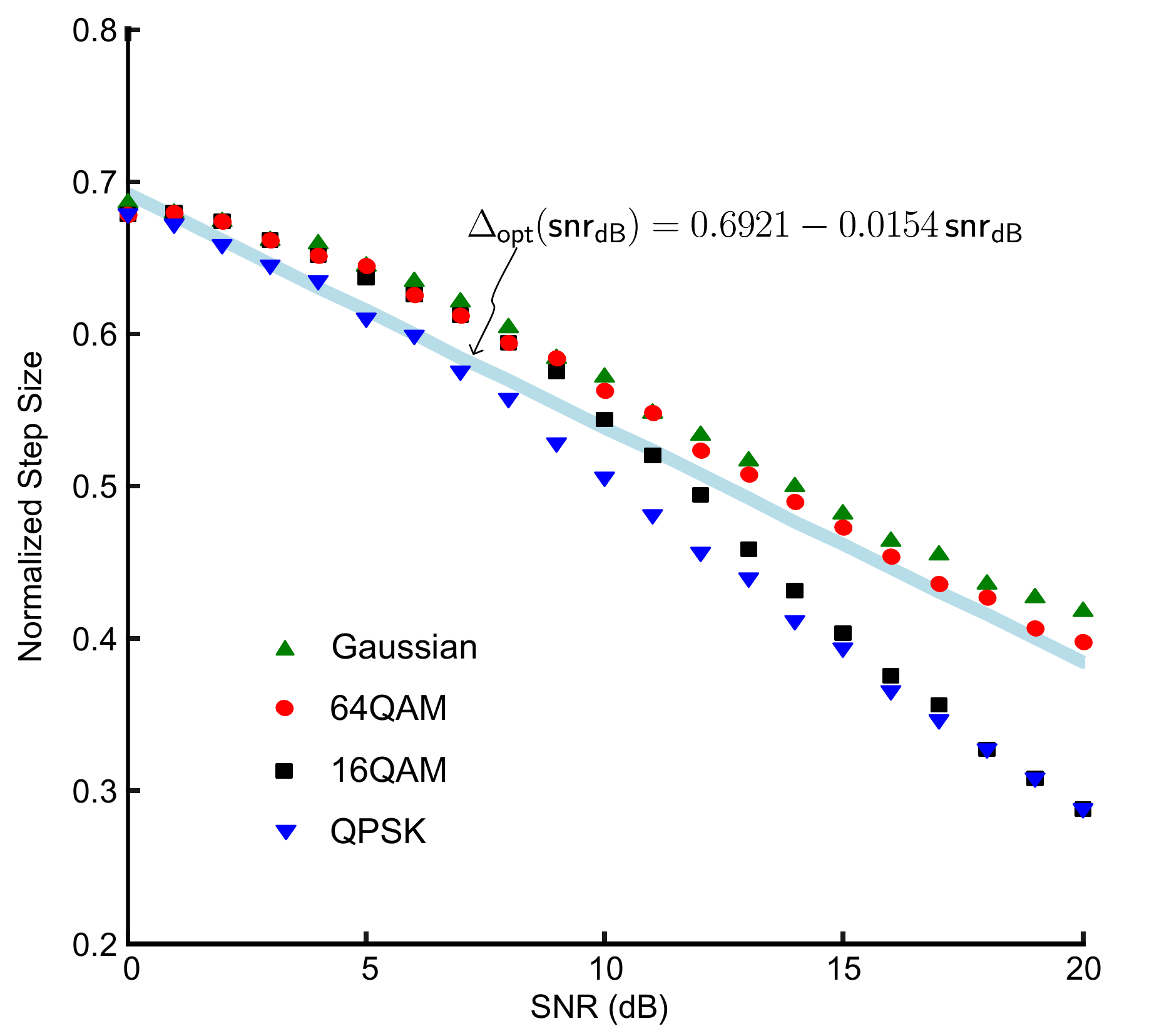} }%
\caption{The optimal step size (normalized by $\Ex\{ |Y_{nt}|^2 \}$).}\label{fig:OptStepSize}
\end{center}
\end{figure}

In the 1 bit ADC (i.e., $\sfB =1$), the quantization output is assigned the value $\frac{\Delta}{2}$ if the input is a positive number and $-\frac{\Delta}{2}$ otherwise. For the Bayes-optimal estimator, the performances are \emph{irrelevant} to any particular value of $\Delta$.\footnote{This property is invalid for other estimators, such as linear estimators \cite{Zhang-15Arxiv}. } This property can be easily achieved by reviewing the likelihood in (\ref{eq:lnkelihood_each}), wherein ${r_{b} = \{ -\infty, \, 0, \infty \}}$ for $b=0,1,2$. Notably, $\Delta$ is not involved at the beginning of data estimation. Therefore, we shall focus on cases with $\sfB > 1$.

Notably, $Y_{nt}$ is the input signal to the quantizer. Direct application of the central limit theorem results in that $Y_{nt}$ can be approximated as a Gaussian distribution with variance $\Ex\{ |Y_{nt}|^2 \} = 1 + \sigma_{w}^2$. For a Gaussian signal with unit variance, the optimal step size that can be used to minimize quantization distortion is computed in \cite{Dhahir-TSP96} and is ${1.008/\sqrt{2} \approx 0.7128}$ if $\sfB=2$.\footnote{The optimal step size \cite{Dhahir-TSP96} is divided by $\sqrt{2}$ in this study because the signal power of the real or imaginary part is $1/\sqrt{2}$.} Under the same setting, that is, $\alpha = N/K = 4$, $\beta = T/K = 10$, $\beta_{\sft} = T_{\sft}/K = 1$, Figure \ref{fig:StepSize_QPSK16QAM_SER} shows the SERs of the Bayes-optimal estimator as a function of the \emph{normalized} step size $\Delta/\sqrt{\Ex\{ |Y_{nt}|^2 \}}$ for $\sfB=2$. We observed that the step size optimized in terms of the SER for the Bayes-optimal estimator is quite different from that for minimizing its distortion.

Figure \ref{fig:OptStepSize} shows the optimal step sizes for different input signals $\qX_{\sfd}$, including QPSK, 16QAM, 64AM, and Gaussian inputs. The optimal step size varies slightly for different input signals and decreases with the increase in SNR. We observe from other simulations that the optimal step size varies only slightly for different settings of $\alpha$ and $\beta$. Thus, we conclude that the optimal step size for the Bayes-optimal estimator is mainly dominated by the SNR.

\begin{table}[t]
\caption{The coefficients $a_0$ and $a_1$ of $\Delta_{\sf opt}({\sf snr}_{\sf dB})$ for $\sfB=2,3,4$. } 
\centering 
\begin{tabular}{lcc}
\hline
  $\sfB$ \hspace{0.5cm} & $a_0$ & $a_1$ \\
\hline \hline
  $2$ \hspace{0.5cm} & $0.6921$ & $-0.0154$ \\
\hline
  $3$ \hspace{0.5cm} & $0.4364$ & $-0.0118$\\
\hline
  $4$ \hspace{0.5cm} & $0.2559$ & $-0.0071$\\
\hline
\end{tabular}
\label{table:CoefficientsOfDelta}
\end{table}

We fit the optimal step sizes for different input signals by using a first-degree polynomial equation
\begin{equation}
\Delta_{\sf opt}({\sf snr}_{\sf dB}) = a_0 + a_1 {\rm snr}_{\sf dB},
\end{equation}
where ${\sf snr}_{\sf dB}$ represents the SNR in dB scale to derive a general expression.
The (least squares fit) coefficients $a_0,a_1$ are listed in Table \ref{table:CoefficientsOfDelta}. The optimal step sizes determined by using $\Delta_{\sf opt}$ are also indicated by the shaded area in Figure \ref{fig:StepSize_QPSK16QAM_SER}. We observe that, although $\Delta_{\sf opt}$ is nonoptimal for each specific input, their corresponding performances remain acceptable. Following the same argument, we derive the corresponding polynomial equation $\Delta_{\sf opt}({\sf snr}_{\sf dB})$ for different quantization bits, with their coefficients listed in Table \ref{table:CoefficientsOfDelta}.

\subsection*{C. Effects due to the Absence of CSIR}

Comparing Figures \ref{fig:SER_QPSK}(a) and \ref{fig:SER_QPSK}(b), we observe that the loss due to no CSIR is small for the Bayes-optimal JCD estimator. Then, we discuss the performances of the Bayes-optimal JCD estimator \emph{with} and \emph{without} the perfect CSIR under various system settings to obtain a better understanding on the effects of CSIR over the quantized MIMO system. In contrast to the QPSK signals used in previous simulations, we focus on the Gaussian inputs, that is, $X_{\sfd} \sim \calN_{\bbC}(0,1)$, in the subsequent experiments. The other system parameters are the same as  that in the previous experiment, that is, ${\alpha = N/K = 4}$, ${\beta = T/K = 10}$, and ${\beta_{\sft} = T_{\sft}/K = 1}$. Figure~\ref{fig:MSEX_QvsUnQ} shows the asymptotic MSE $\mse_{X_{\sfd}}$ for the Bayes-optimal JCD estimator with and without perfect CSIR. The MSE for the pilot-only scheme is also shown. Notably, the Bayes-optimal JCD estimator shows a significant improvement over the pilot-only scheme in the $1$~bit and unquantized cases. The gap between the Bayes-optimal JCD estimator with and without perfect CSIR is small in the unquantized case, whereas the gap is large in the case of the $1$~bit quantizer. By observing the $1$~bit and unquantized cases, we can expect that the gap decreases with the increase in the ADC resolution.

\begin{figure}
\begin{center}
\resizebox{2.8in}{!}{%
\includegraphics*{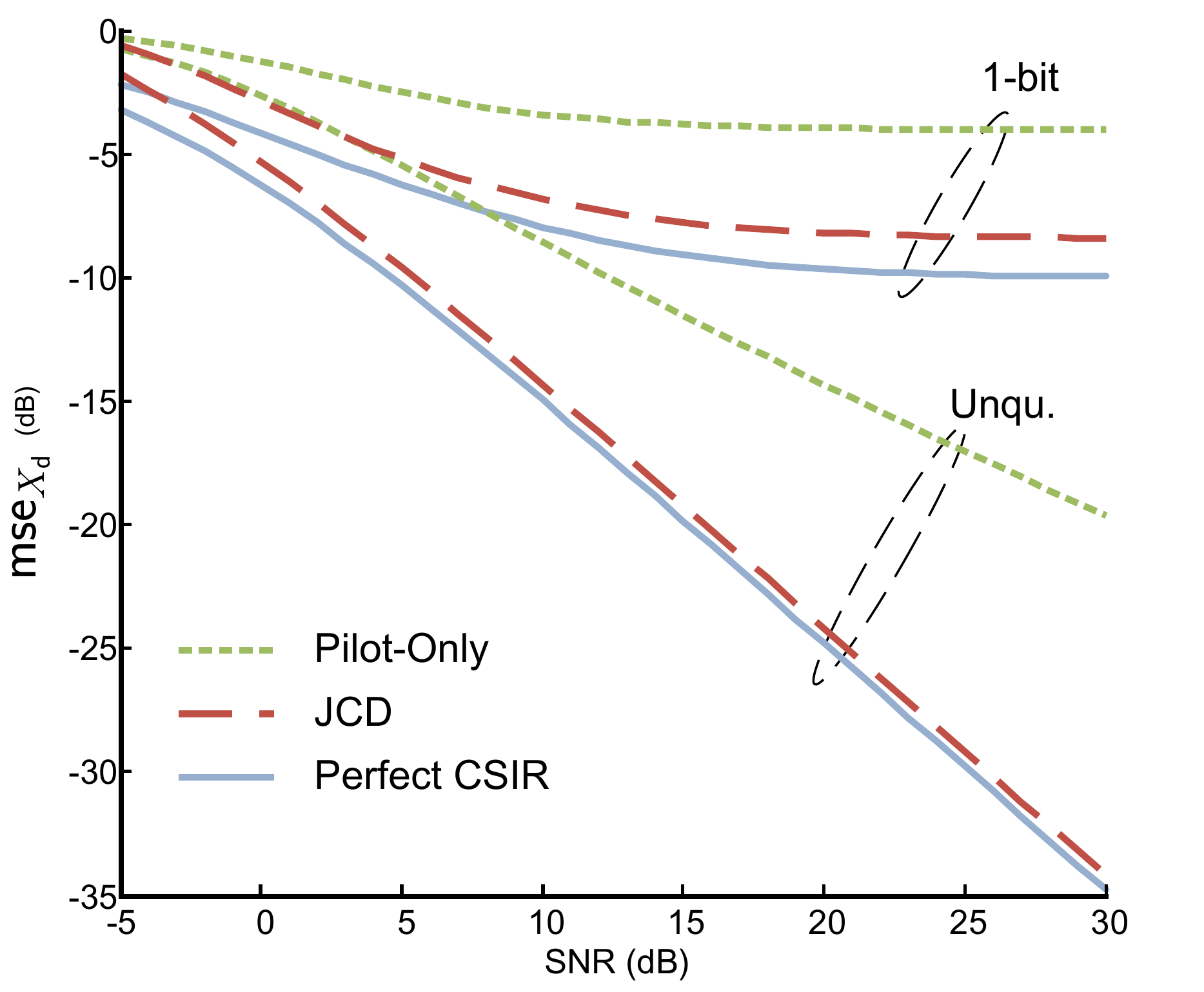} }%
\caption{$\mse_{X_{\sfd}}$ versus SNR for the pilot-only scheme and the Bayes-optimal JCD estimator with and without perfect CSIR under the $1$-bit quantization and unquantized receivers. $\alpha = N/K = 4$, $\beta = T/K = 10$, $\beta_{\sft} = T_{\sft}/K = 1$, and $X_{\sfd,kt} \sim \calN_{\bbC}(0,1)$.}\label{fig:MSEX_QvsUnQ}
\end{center}
\end{figure}

\begin{figure}
\begin{center}
\resizebox{2.8in}{!}{%
\includegraphics*{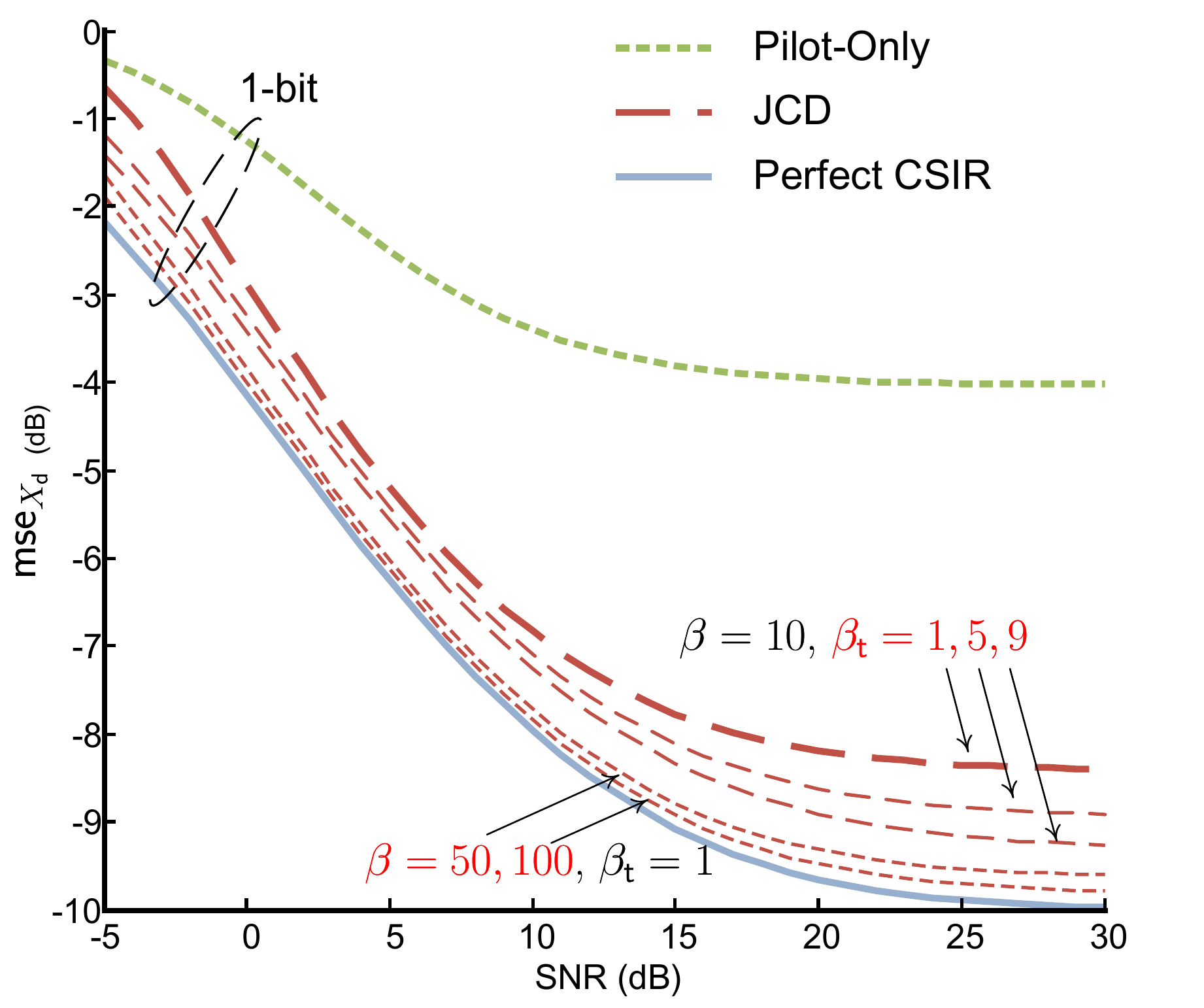} }%
\caption{$\mse_{X_{\sfd}}$ versus SNR for the Bayes-optimal JCD estimator with $1$-bit receivers under various setting of $\beta$ and $\beta_{\sft}$. $\alpha= 4$ and $X_{\sfd,kt} \sim \calN_{\bbC}(0,1)$.}\label{fig:MSEX_All_JCD}
\end{center}
\end{figure}

A straightforward method to reduce the gap of the $1$~bit case is increasing the training length. We provide the MSE results for $\beta_{\sft} = 5$ and $\beta_{\sft} = 9$ in Figure~\ref{fig:MSEX_All_JCD} to verify this argument. However, the improvement achieved by increasing the training length is limited even if $\beta_{\sft} =9$, leaving only $\beta_{\sfd} =1$ for data. Alternatively, we may consider a larger $\beta$, but the maximum $\beta$ is limited by the coherence time. If $\beta = 100$ and $\beta_{\sft} =1$, then the performance of the Bayes-optimal JCD estimator without perfect CSIR is approximately similar to the (fundamental limit) case with perfect CSIR.   Nevertheless, such a long coherence time is usually unavailable in practice.

\section*{VI. Conclusion}
We developed a framework for studying the \emph{best} possible estimation performance of the quantized MIMO system, namely, the massive MIMO system with very low-precision ADCs. In particular, we used the Bayes-optimal inference for the JCD estimation and achieve this estimation by applying the BiG-AMP technique. The asymptotic performances (e.g., MSEs and SERs) w.r.t.~the channels and the payload data were derived and shown as simply characterized by scalar AWGN channels. Monte-Carlo simulations were conducted to demonstrate the accuracy of our analytical results.

The high accuracy of the analytical expressions enable us to quickly and efficiently assess the performance of the quantized MIMO system. Thus, we obtained the following useful guidelines for the system design:
\begin{itemize}
\item We showed that the asymptotic MSE of the channel estimate in the conventional pilot-only scheme decreases by approximately $6$\,dB for each bit added to the ADCs or each doubling of training length. This finding supports the importance of the JCD technique, especially in the \emph{quantized} MIMO system.

\item The optimal step size for minimizing SERs of the Bayes-optimal estimator were shown to be highly different from that for minimizing the distortion of a Gaussian signal and, fortunately, can be quickly determined by our analytical expressions.

\item The Bayes-optimal estimator already exhibits the best possible estimation performance. Even so, we showed that the performance gap between the Bayes-optimal JCD estimator with and without perfect CSIR still cannot be negligible in the \emph{quantized} MIMO system. We also discussed the ways to reduce the gap and then concluded that achieving the same performance as the full CSIR case in the quantized MIMO system is very difficult.
\end{itemize}

Many potential directions for future work are available. The GAMP-based JCD algorithm presented in this paper is a first step toward achieving the \emph{optimal} JCD estimate under the quantized MIMO system. The computational complexity of the GAMP-based JCD algorithm may still be too high to be affordable in a commercial system. One possible solution is to adopt other suboptimal schemes such as linear estimators. Another feasible solution is using mixed-ADC receiver architecture \cite{Liang-15ArXiv} wherein a small number of high-resolution ADCs is available. Thus, CSIR gains high accuracy and facilitates the JCD procedure. For a development in this direction, see \cite{Zhang-15Arxiv}.

\section*{ Appendix A: Derivations of (\ref{eq:hatZ_RealGaussian}) and (\ref{eq:mseZ_RealGaussian})}
In this appendix, we derive the expressions (\ref{eq:hatZ_RealGaussian}) and (\ref{eq:mseZ_RealGaussian}), by applying the techniques in \cite[Chapter
3.9]{Rasmussen-06BOOK}. The derivations below are only dedicated for the real part of the estimator because the imaginary part of the estimator can be
obtained analogously. Note that the signal power and noise power are $v^{p}/2$ and $\sigma_{w}^2/2$, respectively, per real and imaginary part. For ease of
notation, we have abused $\wtY$, $Z$, and $\hp$ to denote ${\rm Re}(\wtY)$, ${\rm Re}(Z)$, and ${\rm Re}(\hp)$, respectively.

To get (\ref{eq:hatZ_RealGaussian}), we begin by deriving the denominator of (\ref{eq:margPost_z}). First, recall from (\ref{eq:lnkelihood_each}) that if $\wtY \in (r_{b-1},r_{b}]$ and $\wtY \leq 0$, the likelihood is given by
\begin{align}
    \sfP_{\sf out}( \wtY | Z )
    &= \Phi{\left( \frac{r_b-Z}{\sqrt{\sigma_{w}^2/2}}\right)} - \Phi{\left( \frac{r_{b-1}-Z}{\sqrt{\sigma_{w}^2/2}}\right)}.\label{eq:Pout_App}
\end{align}
Note that for the special case $b = 1$, we have $r_0 = -\infty$, and the second term of (\ref{eq:Pout_App}) will disappear. Substituting (\ref{eq:Pout_App}) into
the denominator of (\ref{eq:margPost_z}), it can be shown that\footnote{The calculation can be done by using the Gaussian reproduction property given by
footnote 5.}
\begin{multline} \label{eq:normConst}
\int {\sfP_{\sf out}( \wtY | z )} \calN(z; \hp, v^{p}/2)\,\rmd z\\
 = \Phi{\left( \frac{\sign(\wtY)\hp-|r_b|}{\sqrt{(\sigma_{w}^2 + v^{p})/2}}\right)} - \Phi{\left( \frac{\sign(\wtY)\hp-|r_{b-1|}}{\sqrt{(\sigma_{w}^2 + v^{p})/2}}\right)} \triangleq C.
\end{multline}
Differentiating w.r.t.~$\hp$ on both sides of (\ref{eq:normConst}) yields
\begin{multline}  \label{eq:difNormConst}
\int{\left( \frac{z-\hp}{v^{p}/2} \right)} {\sfP_{\sf out}( \wtY | z )} \calN(z; \hp, v^{p}/2)\,\rmd z\\
 =\frac{\sign(\wtY)}{\sqrt{(\sigma_{w}^2 + v^{p})/2}} \Bigg( \phi{\left( \frac{\sign(\wtY)\hp-|r_b|}{\sqrt{(\sigma_{w}^2 + v^{p})/2}} \right)}\\
- \phi{\left( \frac{\sign(\wtY)\hp-|r_{b-1}|}{\sqrt{(\sigma_{w}^2 + v^{p})/2}}\right)} \Bigg),
\end{multline}
where we have used ${\partial \Phi(x)/\partial \hp = \phi(x) \partial x /\partial \hp}$.
Using (\ref{eq:normConst}), (\ref{eq:difNormConst}) can be rearranged as
\begin{multline}  \label{eq:zPoutN}
  \int{z\,\sfP_{\sf out}( \wtY | z )} \calN(z; \hp, v^{p}/2) \,\rmd z\\
  = \hp C + \frac{\sign(\wtY)v^{p}}{\sqrt{2(\sigma_{w}^2 + v^{p})}} \left( {\phi(\eta_1)} - {\phi(\eta_2)} \right),
\end{multline}
where $\eta_2$ and $\eta_1$ are given by (\ref{eq:eta_def}). Multiplying both sizes by $1/C$, we obtain  the marginal posterior mean given in (\ref{eq:hatZ_RealGaussian}).

Similarly, (\ref{eq:mseZ_RealGaussian}) can be calculated by differentiating (\ref{eq:normConst}) twice as
\begin{align}  \label{eq:dif2NormConst}
&\int{\left( \frac{z^2}{(v^{p}/2)^2} - \frac{2\hp z}{(v^{p}/2)^2} + \frac{\hp^2}{(v^{p}/2)^2}
   - \frac{1}{v^{p}/2} \right)} \nonumber \\
& \hspace{3.5cm} \times {\sfP_{\sf out}( \wtY | z )} \calN(z; \hp, v^{p}/2)\,\rmd z \nonumber \\
= &\, \frac{-1}{(\sigma_{w}^2 + v^{p})/2} \Big( \eta_1 {\phi(\eta_1)}  - \eta_2 {\phi(\eta_2)} \Big),
\end{align}
which then can be rearranged as
\begin{multline}  \label{eq:dif2NormConst2}
  \Ex{\left\{ Z^2 \big|\,\hp, v^{p}/2\right\}} = 2\hp \,\Ex{\left\{  Z  \big|\,\hp, v^{p}/2\right\}} + \left({v^{p}/2}-\hp^2\right)\\
 - \frac{1}{C} \frac{(v^{p})^2 }{2(\sigma_{w}^2 + v^{p})} \Big( \eta_1
 {\phi(\eta_1)}  - \eta_2 {\phi(\eta_2)} \Big).
\end{multline}
We also note that
\begin{equation} \label{eq:varZ_def}
\Varx{\left\{ Z \big|\,\hp, v^{p}/2\right\}}  = \Ex{\left\{ Z^2 \big|\,\hp, v^{p}/2\right\}}  - \left(\Ex{\left\{  Z  \big|\,\hp, v^{p}/2\right\}} \right)^2.
\end{equation}
After plugging (\ref{eq:dif2NormConst2}) and (\ref{eq:hatZ_RealGaussian}) into (\ref{eq:varZ_def}), we obtain (\ref{eq:mseZ_RealGaussian}). In the above derivations, we have assumed $\wtY \leq 0$. For $\wtY > 0$, the above derivations can be used in the same way.

\section*{ Appendix B: Proof of Proposition \ref{Pro_FreeEntropy}}
Using the replica method, we first compute the replicate partition function $\Ex_{\wtqY}\left\{\sfP^{\tau}(\qY)\right\}$ in (\ref{eq:LimF}), which with the definition of (\ref{eq:posteriorPr}) can be expressed as
\begin{equation} \label{eq:sf_E1}
    \Ex_{\wtqY}{\left\{\sfP^{\tau}(\wtqY)\right\}}
    = \Ex_{\calqH,\calqX}{\left\{\int\!\rmd \wtqY
    \prod_{a=0}^{\tau} {\sfP_{\sf out}\left(\wtqY \left|\qZ^{(a)}\right.\right)} \right\}},
\end{equation}
where we define ${\qZ^{(a)} \triangleq \qH^{(a)} \qX^{(a)}/\sqrt{K}}$ with $\qH^{(a)}$ and $\qX^{(a)}$ being the $a$-th replica of $\qH$ and $\qX$, respectively, and $\calqX \triangleq \{ \qX^{(a)}, \forall a \}$ and $\calqH \triangleq \{ \qH^{(a)}, \forall a \}$. Here, $(\qH^{(a)},\qX^{(a)})$ are random matrices taken from the distribution $(\sfP_{\sfH},\sfP_{\sfX})$ for $a=0,1, \dots, \tau$. In addition, $\int\!\rmd \wtqY$ denotes the integral w.r.t.~a discrete measure because the quantized output $\wtqY$ is a finite set. We will calculate the right-hand side of (\ref{eq:sf_E1}), by applying the techniques in \cite{Krzakala-13ISIT,Kabashima-14ArXiv,Wen-15ICC} after additional manipulations.

To average over $(\calqH,\calqX)$, we introduce two $(\tau+1)$$\times(\tau+1)$ matrices $\qQ_{H} = [Q_{H}^{ab}] $ and $\qQ_{X_{o}} =[Q_{X_{o}}^{ab}] $ for $o \in \{ \sft,\sfd\}$ whose elements are defined by $Q_{H}^{ab} = \qh_{n}^{(b)} (\qh_{n}^{(a)})^{\dag}/K $ and $Q_{X_{o}}^{ab} = (\qx_{t}^{(a)})^{\dag}\qx_{t}^{(b)}/K$ for $t \in \calT_{o}$, where, $\qh_{n}^{(a)}$ is the $n$th row vector of $\qH^{(a)}$, and $\qx_{t}^{(a)}$ is the $t$th column vector of $\qX^{(a)}$ for $t \in\calT_{\sft}$ or $\calT_{\sfd}$. The definitions of $\qQ_{H}$ and $\qQ_{X_{\sfd}}$ are equivalent to
\begin{align*} 
1 &=  \int \prod_{n=1}^{N} \prod_{0\leq a \leq b}^{\tau}\delta{\left(
 \qh_{n}^{(b)} (\qh_{n}^{(a)})^{\dag} - K Q_{H}^{ab} \right)} \rmd Q_{H}^{ab}, \\
1 &= \int \prod_{o \in \{ \sft,\sfd\}}\prod_{t \in \calT_o}\prod_{0\leq a \leq b}^{\tau}\delta{\left(
 (\qx_{t}^{(a)})^{\dag} \qx_{t}^{(b)} - K Q_{X_{o}}^{ab} \right)}  \rmd Q_{X_{o}}^{ab}, 
\end{align*}
where $\delta(\cdot)$ denotes Dirac's delta. Let $\calqQ_X \triangleq \{ \qQ_{X_o}, \forall o \}$ and $\calqZ  \triangleq \{ \qZ^{(a)}, \forall a\}$. Inserting the
above into (\ref{eq:sf_E1}) yields
\begin{equation}\label{eq:sf_E2}
 \Ex_{\wtqY}\{\sfP^{\tau}(\wtqY)\} = {\int e^{K^2{\cal G}^{(\tau)}}\rmd\mu_H^{(\tau)}(\qQ_H) \rmd\mu_X^{(\tau)}(\calqQ_X)},
\end{equation}
where
\begin{align*} 
{\cal G}^{(\tau)}(\calqQ_Z)&\triangleq \frac{1}{K^2} \log \Ex_{\calqZ }{\left\{  \int\!\rmd\wtqY\prod_{a=0}^{\tau} {\sfP_{\sf out}\left(\wtqY \left| \qZ^{(a)} \right.\right)}  \right\}}, \\
\mu_H^{(\tau)}(\qQ_H) &\triangleq\Ex_{\calqH}{\left\{\prod_{n=1}^{N} \prod_{0\leq a \leq b}^{\tau} \delta{\left(
 \qh_{n}^{(b)} ( \qh_{n}^{(a)} )^{\dag} - K Q_{H}^{ab} \right)}\right\}},\\
\mu_X^{(\tau)}(\calqQ_X) &\triangleq\Ex_{\calqX}{\left\{ \prod_{o,t \in \calT_o}\prod_{0\leq a \leq b}^{\tau} \delta{\left(
 (\qx_{t}^{(a)})^{\dag} \qx_{t}^{(b)} - K Q_{X_{o}}^{ab} \right)}\right\}}.
\end{align*}
Using the Fourier representation of the $\delta$ function via auxiliary matrices ${\tilde{\qQ}_{H} = [\tilde{\qQ}_{H}^{ab}]} \in \bbC^{(\tau+1)\times(\tau+1)}$, $\tilde{\calqQ}_X \triangleq\{ \tilde{\qQ}_{X_{o}} = [\tilde{\qQ}_{X_{o}}^{ab}] \in \bbC^{(\tau+1)\times(\tau+1)}, \forall o \} $ and performing the saddle point method for the integration over $(\qQ_H,\calqQ_X)$, we attain
\begin{equation} \label{eq:sf_E3}
    \frac{1}{K^2}\Ex_{\wtqY}\{\sfP^{\tau}(\wtqY)\} = \Extr_{\qQ_H,\calqQ_X,\tilde{\qQ}_H,\tilde{\calqQ}_X} \Big\{ \calF^{(\tau)} \Big\}
\end{equation}
with
\begin{subequations} \label{eq:saddlePoint}
\begin{align}
& \calF^{(\tau)} \triangleq \nonumber \\
& \frac{1}{K^2} \log \Ex_{\qZ}\left\{  \prod_{n,o,t \in \calT_o}  \int\!\rmd \wtY_{nt} \prod_{a} {\sfP_{\sf out}\left(\wtY_{nt} \left| Z_{nt}^{(a)} \right.\right)}  \right\}  \label{eq:saddlePoint1} \\
& +\frac{1}{K^2} \log \calM_{H}^{(\tau)}(\qQ_{H})  -\alpha {\tr\left(\tilde\qQ_{H}\qQ_{H}\right)} \label{eq:saddlePoint2} \\
& +\frac{1}{K^2}  \log  \calM_{X}^{(\tau)}(\tilde{\calqQ}_X) -{\sum_{o}\beta_o
 \tr\left(\tilde\qQ_{X_{o}}\qQ_{X_{o}}\right)}, \label{eq:saddlePoint3}
\end{align}
\end{subequations}
where $\Extr_{x}\{ f(x) \}$ denotes the extreme value of $f(x)$ w.r.t.~$x$;
\begin{align*}
 &\calM_{H}^{(\tau)}(\tilde{\qQ}_{H}) \triangleq \Ex_{\calqH}\left\{\prod_{n=1}^{N} e^{\tr\left(\tilde\qQ_{H}\qH_{n}^H\qH_{n}\right)}\right\}, \\
 &\calM_{X}^{(\tau)}(\tilde{\calqQ}_X) \triangleq \Ex_{\calqX}\left\{\prod_{o \in \{ \sft,\sfd\}}e^{\tr\left(\tilde\qQ_{X_{o}}\qX_{o}^H \qX_{o}\right)}\right\},
\end{align*}
$\qH_{n}^{H} \triangleq [\qh_{n}^{(0)T}\,\qh_{n}^{(1)T}\cdots \qh_{n}^{(\tau)T} ]^{T}$, $\qX_{o} \triangleq [\qx_{o}^{(0)}\,\qx_{o}^{(1)}\cdots \qx_{o}^{(\tau)}
]$. According to (\ref{eq:LimF}), the average free entropy turns out to be
\begin{equation*}
\calF = \lim_{\tau\rightarrow 0}\frac{\partial}{\partial \tau}\Extr_{\qQ_H,\calqQ_X,\tilde{\qQ}_H,\tilde{\calqQ}_X} \Big\{ \calF^{(\tau)} \Big\}.
\end{equation*}

The saddle points of $\calF^{(\tau)}$ can be found by the point of zero gradient w.r.t.~$\{\qQ_{H},\qQ_{X_{o}},\tilde{\qQ}_{H},\tilde{\qQ}_{X_{o}}\}$ but it is still prohibitive to get explicit expressions about the saddle points. Thus, we assume that the saddle points follow the RS form \cite{Kabashima-14ArXiv} as
\begin{subequations}
\begin{align}
    &\qQ_{H} = (c_{H}-q_{H})\qI + q_{H}\qone\qone^T, \label{eq:RS1} \\
    &\tilde\qQ_{H} =  (\tc_{H}-\tq_{H})\qI + \tq_{H}\qone\qone^T, \\
    &\qQ_{X_{o}} = (c_{X_{o}}-q_{X_{o}})\qI + q_{X_{o}}\qone\qone^T, \label{eq:RS3} \\
    &\tilde\qQ_{X_{o}} = (\tc_{X_{o}}-\tq_{X_{o}})\qI + \tq_{X_{o}}\qone\qone^T.
\end{align}
\end{subequations}
In addition, the application of the central limit theorem suggests that $\qz_{nt} \triangleq [ Z_{nt}^{(0)} \, Z_{nt}^{(1)} \cdots Z_{nt}^{(\tau)} ]^T$ are Gaussian random vectors with ${(\tau+1)\times(\tau+1)}$ covariance matrix $\qQ_{Z_t}$. If $t \in \calT_o$, then the $(a,b)$th entry of $\qQ_{Z_o}$ is given by
\begin{equation} \label{eq:defQ}
     (Z_{nt}^{(a)})^* Z_{nt}^{(b)}
     = Q_{H}^{ab} Q_{X_{o}}^{ab} \triangleq Q_{Z_o}^{ab} .
\end{equation}
As such, we set $\qQ_{Z_o} = (c_{H}c_{X_{o}}-q_{H}q_{X_{o}})\qI + q_{H}q_{X_{o}}\qone $, which is equivalent to introducing to the Gaussian random variable $\qz_{nt}$ for $t \in \calT_o$ as
\begin{equation} \label{eq:eqz}
Z_{nt}^{(a)} = \sqrt{ c_{H}c_{X_{o}}-q_{H}q_{X_{o}} } \,u_c^{(a)} + \sqrt{ q_{H}q_{X_{o}} } v_c,~\mbox{for }a=0,\dots \tau,
\end{equation}
where $u_c^{(a)}$ and $v_c$ are independent standard complex Gaussian random variables. With RS, the problem of seeking the extremum w.r.t.~$\{\qQ_{H},\qQ_{X_{o}},\tilde{\qQ}_{H},\tilde{\qQ}_{X_{o}}\}$ is reduced to seeking the extremum over $\{c_{H},q_{H},c_{{X_{o}}},q_{{X_{o}}},\tc_{H},\tq_{H},\tc_{{X_{o}}},\tq_{{X_{o}}} \}$, which can be obtained by equating the corresponding partial derivatives of the RS expression $\calF^{(\tau)}$ to zero.

To this end, we calculate the RS expression of $\calF^{(\tau)}$ by substituting these RS expressions into (\ref{eq:saddlePoint1})--(\ref{eq:saddlePoint3}). First, for (\ref{eq:saddlePoint1}), we substitute (\ref{eq:eqz}) and perform the expectation w.r.t.~$\qZ$ and integration over $ \wtY_{nt}$, to yield
\begin{align}\label{eq:RS_T1}
    &(\ref{eq:saddlePoint1}) = 2 \alpha
    \sum_{o \in \{ \sft,\sfd\}}\beta_o\log \Bigg( \sum_{b=1}^{2^\sfB} \Ex_{v} \Big\{
    \Psi_{b}\left(V_o \right)
    \left( \Psi_{b}\left(V_o \right)
    \right)^{\tau} \Big\} \Bigg),
\end{align}
where we define $V_o \triangleq \sqrt{q_{HX_{o}}} v $ and
 \begin{equation} \label{eq:Psib1}
 \Psi_{b}(V_o)
 \triangleq \Ex_{u} {\left\{ \frac{1}{\sqrt{ 2\pi \sigma_{w}^2}} \int_{r_{b-1}}^{r_b} \rmd y \, e^{-\frac{(\sqrt{2}y - Z_{o} )^2}{2\sigma_{w}^2}} \right\}}
\end{equation}
with $Z_{o} = \sqrt{ c_{HX_{o}}-q_{H}q_{X_{o}} }  u  + V_o$, and $u$ and $v$ being independent \emph{real} standard Gaussian random variables.\footnote{Note that $u_c^{(a)}$ and $v_c$ in (\ref{eq:eqz}) are standard ``complex'' Gaussian random variables. In this paper, we process the real and imaginary parts separately. Therefore, for ease of notation, we have rescaled all the observation outputs $\wty_{n,j}$ and $z_{n,j}^{(a)}$ by $\sqrt{2}$ so that the real and imaginary parts of these random  variables can be regarded as the standard ``real'' Gaussian random variables.} Performing the expectation w.r.t.~$u$, (\ref{eq:Psib1}) can be expressed as (\ref{eq:Psi_def}). Next, we move to the RS calculation of (\ref{eq:saddlePoint2}). Under the RS assumption, the first term of (\ref{eq:saddlePoint2}) can be written as
\begin{multline}
 \frac{1}{K^2}\log\Ex_{\calqH}\Bigg\{\prod_{n=1}^{N} e^{ \left( \sum_{a=0}^\tau\sqrt{\tq_{H_{n}}}\qh_{n}^{(a)}\right)^H
 \left(\sum_{a=0}^\tau\sqrt{\tq_{H_{n}}}\qh_{n}^{(a)}\right) } \\
 \times e^{ (\tc_{H_{n}}-\tq_{H_{n}}) \sum_{a=0}^\tau (\qh_{n}^{(a)})^H\qh_{n}^{(a)}  } \Bigg\}. \label{eq:LT_Rate_Fun1}
\end{multline}
Then we decouple the first quadratic term in the exponent using the Hubbard-Stratonovich transformation and introducing the auxiliary vector $\qy_{H,n} \in \mathbb{C}^{K}$, to rewrite (\ref{eq:LT_Rate_Fun1}) as
\begin{align}\label{eq:LT_Rate_Fun2}
 & \frac{1}{K^2}  \log\Ex_{\calqH}\Bigg\{\int \prod_{n=1}^{N}  \rmD\qy_{H,n}e^{2{\rm Re}\left(\qy_{H,n}^H(\sum_{a} \sqrt{\tq_{H}}\qh_{n}^{(a)})\right)} \nonumber \\
 & \hspace{3.25cm} \times e^{ (\tc_{H}-\tq_{H})\sum_{a}  (\qh_{n}^{(a)})^H\qh_{n}^{(a)} } \Bigg\} \notag \\
&= \alpha \log  \Ex_{H}\!\Bigg\{ \int \rmd Y_{H} e^{ -\left|Y_{H} - \sqrt{\tq_{H}}H \right|^2+\tc_{H}|H|^2 }
 \nonumber \\
 & \qquad \qquad \times  \left(\Ex_{H'}\!\left\{ e^{2\sqrt{\tq_{H}}{\rm Re}(Y_{H}^*H')+ (\tc_{H}-\tq_{H}) |H'|^2 }  \right\}\right)^{\tau} \Bigg\} .
\end{align}
With RS, the second term of (\ref{eq:saddlePoint2}) can be expressed as
\begin{equation} \label{eq:RSQQ}
 -\alpha\Big(  (c_{H}+\tau q_{H}) (\tc_{H}+\tau \tq_{H}) + \tau (c_{H}-q_{H}) (\tc_{H}-\tq_{H}) \Big).
\end{equation}
Similarly, for the first and second terms of (\ref{eq:saddlePoint2}), we have
\begin{align}\label{eq:LT_Rate2_Fun2}
 & \hspace{0.25cm} \sum_{o \in\{ \sft,\sfd \}} \beta_{o} \log \Ex_{X_{o}}\!\Bigg\{ \int \rmd Y_{X_{o}} e^{ -\left|Y_{X_{o}} - \sqrt{\tq_{X_{o}}}X_{o} \right|^2 + \tc_{X_{o}} |X_{o}|^2 } \nonumber \\
 & \hspace{1cm} \times \left(\Ex_{X'_{o}}\!\left\{ e^{ 2\sqrt{\tq_{X_{o}}}{\rm Re}(Y_{X_{o}}^*X'_{o}) + (\tc_{X_{o}}-\tq_{X_{o}}) |X'_{o}|^2 }  \right\}\right)^{\tau} \Bigg\}  \nonumber \\
  &-\sum_{o \in\{ \sft,\sfd \}} \beta_{o} \Big(  (c_{X_{o}}+\tau q_{X_{o}}) (\tc_{X_{o}}
  +\tau \tq_{X_{o}}) \nonumber \\
  & \hspace{2cm} + \tau (c_{X_{o}}-q_{X_{o}}) (\tc_{X_{o}}-\tq_{X_{o}}) \Big).
\end{align}

Putting together (\ref{eq:LT_Rate_Fun2})--(\ref{eq:LT_Rate2_Fun2}), we have the RS expression of $\calF^{(\tau)}$. The parameters $\{c_{H},q_{H},c_{{X_{o}}},q_{{X_{o}}},\tc_{H},\tq_{H},\tc_{{X_{o}}},\tq_{{X_{o}}} \}$ are determined by setting the partial derivatives of $\calF^{(\tau)}$ to zeros. In doing so, as $\tau \rightarrow 0$, it is easy to get that $\tc_{H} =0$, $\tc_{{X_{o}}} = 0$, $c_{H} =\Ex\{|{H}|^2\}$, and $c_{{X_{o}}} = \Ex\{|{X_{o}}|^2\}$. In order to obtain the more meaningful expressions for the other parameters, we introduce two scalar AWGN channels given in (\ref{eq:scalCh}) and their associated parameters in Section IV-A. Equating the partial derivatives of $\calF^{(\tau)}$ w.r.t.~$\{q_{H},q_{{X_{o}}},\tq_{H},\tq_{{X_{o}}} \}$ to zeros, we obtain the fixed-point equations given in (\ref{eq:fxiedPoints}). Finally, taking the partial derivatives of $\calF^{(\tau)}$ at $\tau=0$, and applying the parameters introduced in Section IV-A, we obtain (\ref{eq:FreeEntropyFinal}).

\section*{ Appendix C: Proof of Proposition \ref{Pro_Decoupling}}
Consider the $(n,k)$-th and $(k,t)$-th entries of $\qH$ and $\qX_{\sfd}$, respectively. We will show that the joint moments of the joint distribution of $(H_{nk},X_{\sfd,kt},\whH_{nk},\whX_{\sfd,kt})$ for some indices $(n,k)$ and $(k,t)$ converges to the joint distribution of
\begin{equation}
    \sfP(H) \sfP(Y_{H}|H) \sfP(H|Y_{H})
    \sfP(X_{\sfd}) \sfP(Y_{X_{\sfd}}|X_{\sfd}) \sfP(X_{\sfd}|Y_{X_{\sfd}}),
\end{equation}
independent of $(n,k)$ and $(k,t)$. 
Following \cite{Guo-05IT}, we proceed to calculate the joint moments
\begin{multline}
\Ex\{ {{\rm Re}(H_{nk})^{i_{{\rm R}_{h}}}} {{\rm Im}(H_{nk})^{i_{{\rm I}_{h}}}} {{\rm Re}(\whH_{nk})^{j_{{\rm R}_{h}}}} {{\rm Im}(\whH_{nk})^{j_{{\rm I}_{h}}}}\\
{{\rm Re}(X_{\sfd,kt})^{i_{{\rm R}_{x}}}}{{\rm Im}(X_{\sfd,kt})^{i_{{\rm I}_{x}}}}{{\rm Re}(\whX_{\sfd,kt})^{j_{{\rm R}_{x}}}}{{\rm Im}(\whX_{\sfd,kt})^{j_{{\rm I}_{x}}}}\}
\end{multline}
for arbitrary non-negative integers $i_{{\rm R}_{h}}$, $i_{{\rm I}_{h}}$, $j_{{\rm R}_{h}}$, $j_{{\rm I}_{h}}$, $i_{{\rm I}_{x}}$, $j_{{\rm R}_{x}}$, $j_{{\rm I}_{x}}$, $j_{{\rm I}_{x}}$. To proceed, we define
\begin{align}
    &f_{h}=\sum_{n,k}
    {\left({\rm Re}(H_{nk}^{(0)})\right)^{i_{{\rm R}_{h}}}}
    {\left({\rm Im}(H_{nk}^{(0)})\right)^{i_{{\rm I}_{h}}}} \nonumber \\
    & \hspace{2.5cm}
    \times {\left({\rm Re}(H_{nk}^{(a_{{\rm R}})})\right)^{j_{{\rm R}_{h}}}}
    {\left({\rm Im}(H_{nk}^{(a_{{\rm I}})})\right)^{j_{{\rm I}_{h}}}},
    \nonumber \\
    &f_{x} =\sum_{k,t}
    {\left({\rm Re}(X_{\sfd,kt}^{(0)})\right)^{i_{{\rm R}_{x}}}}
    {\left({\rm Im}(X_{\sfd,kt}^{(0)})\right)^{i_{{\rm I}_{x}}}} \nonumber \\
    & \hspace{2.5cm}
    \times {\left({\rm Re}(X_{\sfd,kt}^{(b_{{\rm R}})})\right)^{j_{{\rm R}_{x}}}}
    {\left({\rm Im}(X_{\sfd,kt}^{(b_{{\rm I}})})\right)^{j_{{\rm I}_{x}}}},
\end{align}
with $a_{\rm R},a_{\rm I} \in \{ 1, \ldots, \tau \}$, $a_{\rm R} \neq a_{\rm I}$ and $b_{\rm R},b_{\rm I} \in \{ 1, \ldots, \tau \}$, $b_{\rm R} \neq b_{\rm I}$. If we define the \emph{generalized} free entropy as
\begin{multline} \label{eq:genFreeEntropy}
\tilde{\calF} = \frac{1}{K^2} \lim_{\tau \rightarrow 0^{+}}
 \frac{\partial^2}{\partial \varepsilon_h \partial \varepsilon_x}\ln \Ex_{\wtqY}{\left\{e^{\varepsilon_h f_{h} \varepsilon_x f_{x}} \sfP^{\tau}(\wtqY)\right\}} \Bigg|_{\varepsilon_h = 0,\varepsilon_x =0},
\end{multline}
it exactly provides the joint moments of interest.

As $\varepsilon_h =0$ and $\varepsilon_x =0$, $\Ex_{\wtqY}\{e^{\varepsilon_h f_{h} \varepsilon_x f_{x}} \sfP^{\tau}(\wtqY)\}$ reduces to $\Ex_{\wtqY}\{\sfP^{\tau}(\wtqY)\}$ given in (\ref{eq:sf_E1}). Therefore, proceeding with the same steps as in Appendix B from (\ref{eq:sf_E1}) to (\ref{eq:sf_E3}), we get
\begin{equation} \label{eq:sft_E3}
    \frac{1}{K^2}\Ex_{\wtqY}{\left\{e^{\varepsilon_h f_{h} \varepsilon_x f_{x}} \sfP^{\tau}(\wtqY)\right\}} = \Extr_{\qQ_H,\calqQ_X,\tilde{\qQ}_H,\tilde{\calqQ}_X} \Big\{ \tilde{\calF}^{(\tau)} \Big\},
\end{equation}
where $\tilde{\calF}^{(\tau)}$ is exactly identical to (\ref{eq:saddlePoint}) while $\calM_{H}^{(\tau)}(\tilde{\qQ}_{H})$ and $\calM_{X}^{(\tau)}(\tilde{\calqQ}_X)$ should be replaced by
\begin{align*}
 &\tilde{\calM}_{H}^{(\tau)}(\tilde{\qQ}_{H}) = \Ex_{\calqH}\Bigg\{e^{\varepsilon_h f_{h}}\prod_{n=1}^{N} e^{\tr\left(\tilde\qQ_{H}\qH_{n}^H\qH_{n}\right)}\Bigg\}, \\
 &\tilde{\calM}_{X}^{(\tau)}(\tilde{\calqQ}_X) = \Ex_{\calqX}\Bigg\{e^{\varepsilon_x f_{x}}\prod_{o \in \{ \sft, \sfd\}} e^{\tr\left(\tilde\qQ_{X_{o}}\qX_{o}^H \qX_{o}\right)}\Bigg\}.
\end{align*}
Thus, except for $\tilde{\calM}_{H}^{(\tau)}(\tilde{\qQ}_{H})$ and $\tilde{\calM}_{X}^{(\tau)}(\tilde{\calqQ}_X)$, the RS expressions for the other parts of $\tilde{\calF}^{(\tau)}$ can be obtained as in Appendix B. We now only need to obtain the RS expressions for $\log \tilde{\calM}_{H}^{(\tau)}(\tilde{\qQ}_{H})$ and $\log \tilde{\calM}_{X}^{(\tau)}(\tilde{\calqQ}_X)$. The generalized free energy in (\ref{eq:genFreeEntropy}) becomes
\begin{align} \label{eq:FinalF_JMoment}
\tilde{\calF}
    &= \int \rmd Y_{H} \rmd Y_{X_{\sfd}}
    ~\Ex_{H}\{ {\left({\rm Re}(H)\right)^{i_{{\rm R}_{h}}}}
    {\left({\rm Im}(H)\right)^{i_{{\rm I}_{h}}}}
    \sfP(Y_{H}|H) \}\notag \\
    & \hspace{2.0cm} \times\underbrace{\Ex_{H}\{ {\left({\rm Re}(H)\right)^{j_{{\rm R}_{h}}}}
    {\left({\rm Im}(H)\right)^{j_{{\rm I}_{h}}}}
    \sfP(H|Y_{H}) \}}_{{\rm Re}(\whH)^{j_{{\rm R}_{h}}} {\rm Im}(\whH)^{j_{{\rm I}_{h}}} }\notag \\
    & \hspace{2.0cm} \times\Ex_{X_{\sfd}}\{ {\left({\rm Re}(X_{\sfd})\right)^{i_{{\rm R}_{x}}}}
    {\left({\rm Im}(X_{\sfd})\right)^{i_{{\rm I}_{x}}}}
    \sfP(Y_{X_{\sfd}}|X_{\sfd})  \}  \notag \\
    & \hspace{2.0cm} \times\underbrace{\Ex_{X_{\sfd}}\{ {\left({\rm Re}(X_{\sfd})\right)^{i_{{\rm R}_{x}}}}
    {\left({\rm Im}(X_{\sfd})\right)^{i_{{\rm I}_{x}}}}
    \sfP(X_{\sfd}|Y_{X_{\sfd}})  \}}_{ {\rm Re}(\whX_{\sfd})^{i_{{\rm R}_{x}}} {\rm Im}(\whX_{\sfd})^{i_{{\rm I}_{x}}} }
\end{align}
which is the joint moments of $(H,X_{\sfd},\whH,\whX_{\sfd})$. Consequently, the joint moment of interest is thus uniquely determined by (\ref{eq:FinalF_JMoment}) due to the Carleman theorem.

\section*{ Appendix D: Proof of Proposition \ref{pro:mseOfH_PilotOnly}}
In this derivation, we consider the case at infinity SNR, i.e. $\sigma_{w}^2 = 0$, and we let ${\sigma_{h}^2 = 1}$ and ${\sigma_{x_{\sft}}^2 = 1}$ without loss of generality. From (\ref{eq:mse_H_PilotOnly}), as $\beta_{\sft} \rightarrow \infty$, we have $\tq_{H} \rightarrow \infty$. An application of the Taylor expansion yields $1-\mse_{H} = (1 + 1/\tq_{H})^{-1} \approx 1 - 1/\tq_{H}$, and thus we have
\begin{equation} \label{eq:approx_1smseH}
\mse_{H} \approx 1/\tq_{H}.
\end{equation}
Let $u = \frac{ \sqrt{2}r_{b}-\sqrt{1-\mse_{H}}v}{\sqrt{ \mse_{H} }}$.
We then evaluate $\chi_{\sft}$ in (\ref{eq:chi_PilotOnly}) by changing the integration variable from $v$ to $u$, which yields
\begin{equation} \label{eq:chi_PilotOnly_2}
 \chi_{\sft} = \frac{c_{\sfB}}{ \sqrt{\mse_{H}(1-\mse_{H})}},
\end{equation}
where
\begin{multline} \label{eq:CB_Val}
c_{\sfB} =
\sum_{b=1}^{2^{\sfB}} \int \frac{e^{-\frac{\left(\sqrt{ \mse_{H} }z-\sqrt{2}r_{b}\right)^2 }{2 (1-\mse_{H})}}}{\sqrt{2 \pi}}
 \\ \times \frac{\left( \phi(z) - \phi{\left(z-\frac{ \sqrt{2}(r_{b}-r_{b-1}) }{\sqrt{ \mse_{H} }} \right)}  \right)^2}
 {\Phi( z ) - \Phi{\left(z-\frac{ \sqrt{2}(r_{b}-r_{b-1}) }{\sqrt{ \mse_{H} }} \right)}}
  \rmd z.
\end{multline}
As $\mse_{H} \rightarrow 0$, $c_{\sfB}$ can be approximated by
\begin{equation}
c_{\sfB} \approx \frac{1}{(2 \pi)^{3/2}}
\sum_{b=1}^{2^{\sfB}} e^{-r_{b}^2} \int{\frac{ e^{-z^2} }{\Phi( z )}} \rmd z,
\end{equation}
which is a quantizer-dependent constant. Using $\tq_{H} = \beta_{\sft}\chi_{\sft}$ given in (\ref{eq:chi_PilotOnly}) and combining (\ref{eq:approx_1smseH}) and (\ref{eq:chi_PilotOnly_2}), we obtain ${\mse_{H} \approx (\beta_{\sft} c_{\sfB})^{-2}}$ or (\ref{eq:mse_HighSNR}) in dB scale, wherein $C_{\sfB} = -20 \log_{10}(c_{\sfB})$. The values of $C_{\sfB}$ in Table \ref{table:Coefficient_CB} are obtained from (\ref{eq:CB_Val}) numerically.

\section*{ Appendix E: A Generalization of Proposition \ref{Pro_FreeEntropy}}

In this Appendix, we extend Proposition \ref{Pro_FreeEntropy} into the case where users have different large-scale fading factors $\sigma_{h_k}^2$. This task can be performed by proceeding with the
same steps as in Appendix A, and the proof is omitted.

Similar to (\ref{eq:scalCh}), we define the scalar AWGN channels for this case:
\begin{subequations} \label{eq:scalCh_MU}
\begin{align}
 Y_{X_{\sfd,k}} &=  \sqrt{\tq_{{X_{\sfd,k}}}} X_{\sfd,k} + W_{X_{\sfd,k}}, \label{eq:scalCh_X_MU} \\
 Y_{H_k} &= \sqrt{\tq_{H_k}} H_k + W_{H_k}, \label{eq:scalCh_H_MU}
\end{align}
\end{subequations}
where $W_{X_{\sfd,k}},W_{H_k} \sim \calN_{\bbC}(0,1)$, $X_{\sfd,k} \sim \sfP_{X_{\sfd}}$, and
$H_k \sim \sfP_{\sfH_k} \equiv \calN_{\bbC}(0,\sigma_{h_k}^2)$.
For ease of notation, we use $\ang{ a_k } = \frac{1}{K} \sum_{k=1}^{K} a_k$
to represent the average over a set $\{ a_k: k =1,\ldots,K \}$.

\begin{Proposition} \label{Pro_FreeEntropy_MU}
As $K \to \infty$, the asymptotic free entropy is
\begin{align}
 \calF &= \alpha \sum_{o \in \{\sft,\,\sfd\} } \beta_{o} \Bigg( \sum_{b=1}^{2^{\sfB}} \int\!\rmD v\,
   \Psi_{b}\left( V_{o} \right)
   \log \Psi_{b}\left( V_{o} \right)
    \Bigg) \notag \\
  & \hspace{-0.35cm} - \alpha  \ang{ I(H_k;Y_{H_k}|\tq_{H_k}) }
  - \beta_{\sfd}  \ang{ I(X_{\sfd,k};Y_{X_{\sfd,k}}|\tq_{X_{\sfd,k}}) } \notag \\
  & \hspace{-0.35cm} + \alpha \ang{ (c_{H_k} - q_{H_k})\tq_{H_k} } + \sum_{o \in \{\sft,\,\sfd\} } \beta_{o} \ang{ (c_{X_{o}}-q_{X_{o,k}}) \tq_{X_{o,k}} }, \label{eq:FreeEntropyFinal_MU}
\end{align}
where
\begin{multline}\label{eq:Psi_def_MU}
 \Psi_{b}(V_{o})
\triangleq \Phi{\left(\frac{ \sqrt{2}r_{b}-V_{o}}{\sqrt{\sigma_{w}^2 + \ang{ c_{H_k}c_{X_{o}}-q_{H_k}q_{X_{\sfd,k}}}}} \, \right)}\\
\quad - \Phi{\left(\frac{ \sqrt{2}r_{b-1}-V_{o}}{\sqrt{\sigma_{w}^2 + \ang{ c_{H_k}c_{X_{o}}-q_{H_k}q_{X_{o,k}}}}} \, \right)};
\end{multline}
${V_{o} \triangleq \sqrt{\ang{q_{H_k}q_{X_{o,k}}}} \,v}$ for $o \in \{\sft,\,\sfd\}$;
$I(H_k;Y_{H_k}|\tq_{H_k})$ is the mutual information between $Y_{H_k}$ and $H_k$;
$I(X_{\sfd,k};Y_{X_{\sfd,k}}|\tq_{X_{\sfd,k}})$ is the mutual information between $Y_{X_{\sfd,k}}$ and $X_{\sfd,k}$; and
$c_{X_{o}} = \sigma_{x_{o}}^2$, $c_{H_k} = \sigma_{h_k}^2$. In (\ref{eq:FreeEntropyFinal_MU}), the other parameters $\{ q_{X_{o,k}}, q_{H_k}, \tq_{X_{o,k}},\tq_{H_k} \}$ are obtained from the solutions to the fixed-point equations
\begin{subequations} \label{eq:fxiedPoints_MU}
\begin{align}
 \hspace{-0.25cm} \tq_{H_k} & = \beta_{\sft,k} q_{X_{\sft,k}} \chi_{\sft} + \beta_{\sfd} q_{X_{\sfd,k}} \chi_{\sfd},
 && \hspace{-0.15cm} q_{H_k} = c_{H_k} - \mse_{H_k},\\
 \hspace{-0.25cm} \tq_{X_{\sft,k}} & = \alpha q_{H_k} \chi_{\sft},
 && \hspace{-0.75cm} q_{X_{\sft,k}} = c_{X_{\sft}} - \mse_{X_{\sft,k}}, \\
 \hspace{-0.25cm} \tq_{X_{\sfd,k}} &= \alpha q_{H_k} \chi_{\sfd},
 &&  \hspace{-0.75cm} q_{X_{\sfd,k}} = c_{X_{\sfd}} - \mse_{X_{\sfd,k}},
\end{align}
\end{subequations}
where $\mse_{X_{\sft,k}} = 0$, and $\mse_{H_k}$ and $\mse_{X_{\sfd,k}}$ are the MSEs of the Bayes-optimal estimators over (\ref{eq:scalCh_H_MU}) and (\ref{eq:scalCh_X_MU}), respectively. In (\ref{eq:fxiedPoints}), we have defined
\begin{equation} \label{eq:chi_def_MU}
 \chi_{o} \triangleq \sum_{b=1}^{2^{\sfB}} \int\!\rmD v \frac{\Big(\Psi'_{b}\left(\sqrt{\ang{q_{H_k}q_{X_{o,k}}}} v \right)\Big)^2}{\Psi_{b}\left(
 \sqrt{\ang{q_{H_k}q_{X_{o,k}}}} v \right)},~\mbox{for }{o \in \{\sft,\,\sfd\}}
\end{equation}
with $\Psi_{b}(\cdot)$ given by (\ref{eq:Psi_def_MU}) and $\Psi'_{b}(V_{o}) = \frac{\partial \Psi_{b}(V_{o})}{\partial V_{o}}$.
\end{Proposition}

{\renewcommand{\baselinestretch}{1.1}
}


\begin{thebibliography}{10}
\providecommand{\url}[1]{#1} \csname url@samestyle\endcsname \providecommand{\newblock}{\relax} \providecommand{\bibinfo}[2]{#2}
\providecommand{\BIBentrySTDinterwordspacing}{\spaceskip=0pt\relax} \providecommand{\BIBentryALTinterwordstretchfactor}{4}
\providecommand{\BIBentryALTinterwordspacing}{\spaceskip=\fontdimen2\font plus \BIBentryALTinterwordstretchfactor\fontdimen3\font minus
  \fontdimen4\font\relax}
\providecommand{\BIBforeignlanguage}[2]{{%
\expandafter\ifx\csname l@#1\endcsname\relax
\typeout{** WARNING: IEEEtran.bst: No hyphenation pattern has been}%
\typeout{** loaded for the language `#1'. Using the pattern for}%
\typeout{** the default language instead.}%
\else \language=\csname l@#1\endcsname \fi #2}} \providecommand{\BIBdecl}{\relax} \BIBdecl

\bibitem{Wang-14COM-Mag}
{C.-X. Wang {\em et al}.}, ``{Cellular architecture and key technologies for 5G
  wireless communication networks},'' \emph{IEEE Commun. Mag.}, vol.~52, no.~2,
  pp. 122--130, Feb. 2014.

\bibitem{Marzetta-10TW}
{T. L. Marzetta}, ``{Noncooperative cellular wireless with unlimited numbers of
  base station antennas},'' \emph{IEEE Trans. Wireless Commun.}, vol.~9,
  no.~11, pp. 3590--3600, Nov. 2010.

\bibitem{Larsson-14COMMag}
{E. G. Larsson {\em et al}.}, ``{Massive MIMO for next generation wireless
  systems},'' \emph{IEEE Commun. Mag.}, vol.~52, no.~2, pp. 186--195, Feb.
  2014.

\bibitem{Andrews-14JSAC}
{J. G. Andrews {\em et al}.}, ``{What will 5G be?}'' \emph{IEEE J. Sel. Areas
  Commun.}, vol.~32, no.~6, pp. 1065--1082, June 2014.

\bibitem{Walden-99JSAC}
{R. H. Walden}, ``{Analog-to-digital converter survey and analysis},''
  \emph{IEEE J. Sel. Areas Commun.}, vol.~17, no.~4, pp. 539--550, Apr. 1999.

\bibitem{Singh-09TCOM}
{J. Singh, O. Dabeer, and U. Madhow}, ``{On the limits of communication with
  low-precision analog-to-digital conversion at the receiver},'' \emph{IEEE
  Trans. Commun.}, vol.~57, no.~12, pp. 3629--3639, Dec. 2009.

\bibitem{Koch-10}
{T. Koch and A. Lapidoth}, ``{Increased capacity per unit-cost by
  oversampling},'' in \emph{Proc. IEEE 26th Conv. Elect. Electron. Eng.
  Israel}, Eilat, Israel, Nov. 17--20 2010, pp. 684--688.

\bibitem{Zhang-12TCOM}
{W. Zhang}, ``{A general framework for transmission with transceiver distortion
  and some applications},'' \emph{IEEE Trans. Commun.}, vol.~60, no.~2, pp.
  384--399, Feb. 2012.

\bibitem{Wang-13TCOM}
{Z. Wang, H. Yin, W. Zhang, and G. Wei}, ``{Monobit digital receivers for QPSK:
  design, performance and impact of IQ imbalances},'' \emph{IEEE Trans.
  Commun.}, vol.~61, no.~8, pp. 3292--3303, Aug. 2013.

\bibitem{Mezghani-08ISIT}
{A. Mezghani and J. Nossek}, ``{Analysis of Rayleigh-fading channels with 1-bit
  quantized output},'' in \emph{Proc. IEEE Int. Symp. Inf. Theory (ISIT)},
  Toronto, Canada, Jul. 2008, pp. 260--264.

\bibitem{Mo-15TSP}
{J. Mo and R. W. Heath Jr}, ``{Capacity analysis of one-bit quantized MIMO
  systems with transmitter channel state information},'' \emph{IEEE Trans.
  Signal Process.}, vol.~63, no.~20, pp. 5498--5512, Oct. 2015.

\bibitem{Liang-15ArXiv}
\BIBentryALTinterwordspacing {N. Liang and W. Zhang}, ``{Mixed-ADC massive MIMO},'' preprint, 2015.
  [Online]. Available: \url{http://arxiv.org/abs/1504.03516}
\BIBentrySTDinterwordspacing

\bibitem{Bai-15TETT}
{Q. Bai and J. Nossek}, ``{Energy efficiency maximization for 5G multi-antenna
  receivers},'' \emph{Trans. Emerging Telecommun. Technol.}, vol.~26, no.~1,
  pp. 3--14, Jan. 2015.

\bibitem{Orhan-15ITA}
{O. Orhan, E. Erkip, and S. Rangan}, ``{Low power analog-to-digital conversion
  in millimeter wave systems: Impact of resolution and bandwidth on
  performance},'' in \emph{Information Theory and Applications Workshop (ITA)},
  San Diego, CA, USA, Feb. 2015, pp. 191--198.

\bibitem{Mo-15ArXiv-Feedback}
{J. Mo and R. W. Heath Jr}, ``{Limited feedback in multiple-antenna systems
  with one-bit quantization},'' in \emph{Asilomar Conf. Signals, Systems and
  Computers}, Pacific Grove, USA, Nov. 2015.

\bibitem{Nakamura-08ISITA}
{K. Nakamura and T. Tanaka}, ``{Performance analysis of signal detection using
  quantized received signals of linear vector channel},'' in \emph{Proc. Inter.
  Symp. Inform. Theory and its Applications (ISITA)}, Auckland, New Zealand,
  Dec. 2008.

\bibitem{Mezghani-10ISIT}
{A. Mezghani and J. Nossek}, ``{Belief propagation based MIMO detection
  operating on quantized channel output},'' in \emph{Proc. IEEE Int. Symp.
  Inform. Theory (ISIT)}, Austin, TX, 13-18 June 2010, pp. 2113--2117.

\bibitem{Risi-14ArXiv}
\BIBentryALTinterwordspacing {C. Risi, D. Persson, and E. G. Larsson}, ``{Massive MIMO with 1-bit ADC},''
  preprint, 2014. [Online]. Available: \url{http://arxiv.org/abs/1404.7736.}
\BIBentrySTDinterwordspacing

\bibitem{Wang-15TWCOM}
{S. Wang, Y. Li, and J. Wang}, ``{Multiuser detection in massive spatial
  modulation MIMO with low-resolution ADCs},'' \emph{IEEE Trans. Wireless
  Commun.}, vol.~14, no.~4, pp. 2156--2168, Apr. 2015.

\bibitem{Jacobsson-15ArXiv}
{S. Jacobsson, G. Durisi, M. Coldrey, U. Gustavsson, and C. Studer}, ``{One-bit
  massive MIMO: Channel estimation and high-order modulations},'' in \emph{2015
  IEEE Int. Conf. Communication Workshop (ICCW)}, London, UK, 8-12 June 2015,
  pp. 1304--1309.

\bibitem{Choi-15TSP}
{J. Choi, D. J. Love, D. R. Brown III, M. Boutin}, ``{Quantized distributed
  reception for MIMO wireless systems using spatial multiplexing},'' \emph{IEEE
  Trans. Signal Process.}, vol.~63, no.~13, pp. 3537--3548, Jul. 2015.

\bibitem{Xu-14JSM}
{Y. Xu, Y. Kabashima, and L. Zdeborov\'{a}}, ``{Bayesian signal reconstruction
  for 1-bit compressed sensing},'' \emph{J. Stat. Mech.}, no.~11, p. P11015,
  2014.

\bibitem{Choi-15Arxiv}
\BIBentryALTinterwordspacing {J. Choi, J. Mo, and R. W. Heath Jr}, ``{Near maximum-likelihood detector and
  channel estimator for uplink multiuser massive MIMO systems with one-bit
  ADCs},'' preprint, 2015. [Online]. Available:
  \url{http://arxiv.org/abs/1507.04452}
\BIBentrySTDinterwordspacing

\bibitem{Studer-15Arxiv}
\BIBentryALTinterwordspacing {C. Studer and G. Durisi}, ``{Quantized massive MU-MIMO-OFDM uplink},''
  preprint, 2015. [Online]. Available: \url{http://arxiv.org/abs/1509.07928}
\BIBentrySTDinterwordspacing

\bibitem{Mezghani-10WSA}
{A. Mezghani, F. Antreich, and J. A. Nossek}, ``{Multiple parameter estimation
  with quantized channel output},'' in \emph{Int. ITG Workshop Smart Antennas
  (WSA)}, Bremen, Feb. 2010, pp. 143--150.

\bibitem{Mo-14ACSSP}
{J. Mo, P. Schniter, N. G. Prelcic, and R. W. Heath Jr}, ``{Channel estimation
  in millimeter wave MIMO systems with one-bit quantization},'' in \emph{Proc.
  Asilomar Conf. Signals, Systems and Computers}, CA, USA, Nov. 2014.

\bibitem{Takeuchi-13TIT}
{K. Takeuchi, R. R. M\"{u}ller, M. Vehkaper\"{a}, and T. Tanaka}, ``{On an
  achievable rate of large Rayleigh block-fading MIMO channels with no CSI},''
  \emph{IEEE Trans. Inf. Theory}, vol.~59, no.~10, pp. 6517--6541, Oct. 2013.

\bibitem{Ma-14TSP}
{J. Ma and L. Ping}, ``{Data-aided channel estimation in large antenna
  systems},'' \emph{IEEE Trans. Signal Processing}, vol.~62, no.~12, pp.
  3111--3124, June 2014.

\bibitem{Wen-15ISIT}
{C.-K. Wen, S. Jin, K.-K. Wong, C.-J. Wang, and G. Wu}, ``{Joint
  channel-and-data estimation for large-MIMO systems with low-precision
  ADCs},'' in \emph{Proc. IEEE Int. Symp. Inform. Theory (ISIT)}, Hong Kong,
  Jun. 2015, pp. 1237--1241.

\bibitem{Parker-14TSP}
{J. T. Parker, P. Schniter, and V. Cevher}, ``{Bilinear generalized approximate
  message passing},'' \emph{IEEE Trans. Signal Process.}, vol.~62, no.~22, pp.
  5839--5853, Nov. 2014.

\bibitem{Hazewinkel-01}
\BIBentryALTinterwordspacing {M. Hazewinkel}, ``{Normal distribution},'' \emph{Encyclopedia of Mathematics,
  Springer}, 2001. [Online]. Available:
  \url{https://www.encyclopediaofmath.org/index.php/Normal\_distribution}
\BIBentrySTDinterwordspacing

\bibitem{Poor-94BOOK}
H.~V. Poor, \emph{{An Introduction to Signal Detection and Estimation}}.\hskip
  1em plus 0.5em minus 0.4em\relax New York: Springer-Verlag, 1994.

\bibitem{Rasmussen-06BOOK}
{C. E. Rasmussen and C. K. I. Williams}, \emph{{Gaussian Processes for Machine
  Learning}}.\hskip 1em plus 0.5em minus 0.4em\relax MIT Press, 2006.

\bibitem{Ziniel-15TSP}
{J. Ziniel, P. Schniter, and P. Sederberg}, ``{Binary linear classification and
  feature selection via generalized approximate message passing},'' \emph{IEEE
  Trans. Signal Process.}, vol.~63, no.~8, pp. 2020--2032, Apr. 2015.

\bibitem{Donoho-09PNAS}
{D. L. Donoho, A. Maleki, and A. Montanari}, ``{Message passing algorithms for
  compressed sensing},'' \emph{Proc. Nat. Acad. Sci.}, vol. 106, no.~45, pp. 18
  914--18 919, 2009.

\bibitem{Rangan-11ISIT}
{S. Rangan}, ``{Generalized approximate message passing for estimation with
  random linear mixing},'' in \emph{Proc. IEEE Int. Symp. Inform. Theory
  (ISIT)}, Saint Petersburg, Russia, Aug. 2011, pp. 2168--2172.

\bibitem{Wu-14JSTSP}
{S. Wu, L. Kuang, Z. Ni, J. Lu, D. Huang, and Q. Guo}, ``{Low-complexity
  iterative detection for large-scale multiuser MIMO-OFDM systems using
  approximate message passing},'' \emph{IEEE J. Sel. Topics Signal Process.},
  vol.~8, no.~5, pp. 902--915, Oct. 2014.

\bibitem{Krzakala-13ISIT}
{F. Krzakala, M. M\'{e}zard, and L. Zdeborov\'{a}}, ``{Phase diagram and
  approximate message passing for blind calibration and dictionary learning},''
  in \emph{Proc. IEEE Int. Symp. Inform. Theory (ISIT)}, Istanbul, Turkey, July
  2013, pp. 659--663.

\bibitem{Kabashima-14ArXiv}
\BIBentryALTinterwordspacing {Y. Kabashima, F. Krzakala, M. M\'{e}zard, A. Sakata, and L. Zdeborov\'{a}},
  ``{Phase transitions and sample complexity in Bayes-optimal matrix
  factorization},'' preprint 2014. [Online]. Available:
  \url{http://arxiv.org/abs/1402.1298.}
\BIBentrySTDinterwordspacing

\bibitem{Nishimori-01BOOK}
H.~Nishimori, \emph{{Statistical Physics of Spin Glasses and Information
  Processing: An Introduction}}.\hskip 1em plus 0.5em minus 0.4em\relax ser.
  Number 111 in Int. Series on Monographs on Physics. Oxford U.K.: Oxford Univ.
  Press, 2001.

\bibitem{Tanaka-02IT}
{T. Tanaka}, ``{A statistical-mechanics approach to large-system analysis of
  CDMA multiuser detectors},'' \emph{IEEE Trans. Inf. Theory}, vol.~48, no.~11,
  pp. 2888--2910, Nov. 2002.

\bibitem{Moustakas-03TIT}
{A. L. Moustakas, S. H. Simon, and A. M. Sengupta}, ``{MIMO capacity through
  correlated channels in the presence of correlated interferers and noise: a
  (not so) large N analysis},'' \emph{IEEE Trans. Inf. Theory}, vol.~49,
  no.~10, pp. 2545--2561, Oct. 2003.

\bibitem{Guo-05IT}
{D. Guo and S. Verd\'{u} }, ``{Randomly spread CDMA: asymptotics via
  statistical physics},'' \emph{IEEE Trans. Inf. Theory}, vol.~51, no.~1, pp.
  1982--2010, Jun. 2005.

\bibitem{Muller-03TSP}
{R. R. M\"{u}ller}, ``{Channel capacity and minimum probability of error in
  large dual antenna array systems with binary modulation},'' \emph{IEEE Trans.
  Signal Process.}, vol.~51, no.~11, pp. 2821--2828, Nov. 2003.

\bibitem{Wen-07IT}
{C.-K. Wen and K.-K. Wong}, ``{Asymptotic analysis of spatially correlated MIMO
  multiple-access channels with arbitrary signaling inputs for joint and
  separate decoding},'' \emph{IEEE Trans. Inf. Theory}, vol.~53, no.~1, pp.
  252--268, Jan. 2007.

\bibitem{Hatabu-09PRE}
{A. Hatabu, K. Takeda, and Y. Kabashima}, ``{Statistical mechanical analysis of
  the Kronecker channel model for multiple-input multipleoutput wireless
  communication},'' \emph{Phys. Rev. E}, vol.~80, pp. 061\,124(1--12), 2009.

\bibitem{Girnyk-14TWC}
{M. A. Girnyk, M. Vehkaper\"{a}, and L. K. Rasmussen}, ``{Large-system analysis
  of correlated MIMO channels with arbitrary signaling in the presence of
  interference},'' \emph{IEEE Trans. Wireless Commun.}, vol.~13, no.~4, pp.
  1536--1276, Apr. 2014.

\bibitem{Bayati-11IT}
{M. Bayati and A. Montanari}, ``{The dynamics of message passing on dense
  graphs, with applications to compressed sensing,},'' \emph{IEEE Trans. Inf.
  Theory}, vol.~57, no.~2, pp. 764--785, Feb. 2011.

\bibitem{Zhang-13JSAC}
{J. Zhang, C.-K. Wen, S. Jin, X. Gao, and K.-K. Wong}, ``{On capacity of
  large-scale MIMO Multiple access channels with distributed sets of correlated
  antennas},'' \emph{IEEE J. Sel. Areas Commun.}, vol.~31, no.~2, pp. 133--148,
  Feb. 2013.

\bibitem{Cover-BOOK}
T.~Cover and J.~A. Thomas, \emph{{Elements of Information Theory}}.\hskip 1em
  plus 0.5em minus 0.4em\relax New York: Wiley, 1991.

\bibitem{Proakis-95}
J.~G. Proakis, \emph{{Digital Communications}}.\hskip 1em plus 0.5em minus
  0.4em\relax McGraw Hill, 4th ed., 1995.

\bibitem{Oppenheim-09BOOK}
A.~V. Oppenheim and R.~W. Schafer, \emph{{Discrete-Time Signal
  Processing}}.\hskip 1em plus 0.5em minus 0.4em\relax Prentice Hall, 3rd ed.,
  2009.

\bibitem{Dhahir-TSP96}
{N. Al-Dhahir and J. M. Cioffi}, ``{On the uniform ADC bit precision and clip
  level computation for a Gaussian signal},'' \emph{IEEE Trans. Signal
  Process.}, vol.~44, no.~2, pp. 434--438, Feb. 1996.

\bibitem{Zhang-15Arxiv}
\BIBentryALTinterwordspacing {T.-C. Zhang, C.-K. Wen, S. Jin, T. Jiang}, ``{Mixed-ADC massive MIMO
  detectors: performance analysis and design optimization},'' preprint, 2015.
  [Online]. Available: \url{http://arxiv.org/abs/1509.07928}
\BIBentrySTDinterwordspacing

\bibitem{Wen-15ICC}
{C.-K. Wen, Y. Wu, K.-K. Wong, R. Schober, and P. Ting}, ``{Performance limits
  of massive MIMO systems based on Bayes-optimal inference},'' in \emph{IEEE
  Int. Conf. Commun. (ICC)}, London, UK, 2015.

\end{thebibliography}
\end{document}